%% file: NSW_SubAd.tex
\documentclass[11pt]{article} 

\usepackage[top=1in, bottom=1in, left=1in, right=1in]{geometry} 

\usepackage{libertine}

\usepackage{comment,xspace}
\usepackage[dvipsnames]{xcolor} 
\definecolor{ptblue}{RGB}{15,76,129} 
\definecolor{ptemerald}{HTML}{009473} 
\definecolor{bluegray}{rgb}{0.4, 0.6, 0.8}
\definecolor{ptilluminating}{HTML}{F5DF4D} 
\definecolor{ptgray}{HTML}{939597} 

\usepackage[ruled,linesnumbered,vlined]{algorithm2e}
\usepackage{algorithmic}
\usepackage{comment,xspace,enumitem}
\usepackage{booktabs}

\usepackage{changepage}
\usepackage{booktabs,multirow,subcaption}

\usepackage[breakable, theorems, skins]{tcolorbox}

\usepackage{amsmath,amsfonts,amssymb,amsthm,mathtools,thmtools} 

\DeclareMathOperator*{\argmax}{arg\,max}

\usepackage[square,numbers,sort]{natbib} 

\usepackage{hyperref}
\hypersetup{
	linktocpage,
	colorlinks=true,
	citecolor=cobalt, 
	urlcolor=ptblue, 
	linkcolor=cobalt, 
}
\usepackage{cleveref} 
\usepackage{crossreftools} 
\usepackage{tikz} 
\theoremstyle{plain}
\newtheorem{theorem}{Theorem}[section]
\newtheorem{corollary}[theorem]{Corollary}
\newtheorem{proposition}[theorem]{Proposition}
\newtheorem{lemma}[theorem]{Lemma}

\newtheorem{claim}[theorem]{Claim}

\theoremstyle{definition}
\newtheorem{definition}[theorem]{Definition}

\newtheorem*{theorem*}{Theorem}

\theoremstyle{remark}
\newtheorem*{remark}{\upshape\bfseries Remark}

\definecolor{cobalt}{rgb}{0.0, 0.28, 0.67}

\newcommand{\ms}[1]{{#1}\xspace}

\makeatletter 
\newcommand{\EF}[1]{\if\relax\detokenize\expandafter{\@firstofone#1{}}\relax EF\xspace\else EF#1\fi}
\makeatother
\newcommand{\EFone}{\textsc{EF1}\xspace}
\newcommand{\EFX}{\textsc{EFx}\xspace}
\newcommand{\PO}{\textsc{PO}\xspace}

\newcommand{\NSW}{\operatorname{NSW}}

\newcommand{\set}[1]{{\left\{#1\right\}}}

\newcommand{\calA}{\mathcal{A}}
\newcommand{\calB}{\mathcal{B}}

\newcommand{\calI}{\mathcal{I}}
\newcommand{\calP}{\mathcal{P}}

\newcommand{\tilA}{\widetilde{\mathcal{A}}}
\newcommand{\tildA}{\widetilde{A}}
\newcommand{\calR}{\mathcal{R}}
\newcommand{\calM}{\mathcal{M}}
\newcommand{\calF}{\mathcal{F}}

\title{\bfseries Compatibility of Fairness and Nash Welfare under \\ Subadditive Valuations}

\author{Siddharth Barman\thanks{Indian Institute of Science, Bangalore;  {barman@iisc.ac.in} }  \and  Mashbat Suzuki\thanks{UNSW, Sydney; mashbat.suzuki@unsw.edu.au}}

\date{}

\begin{document}

\maketitle 

\abstract{
	We establish a compatibility between fairness and efficiency, captured via Nash Social Welfare (NSW), under the broad class of subadditive valuations. We prove that, for subadditive valuations, there always exists a {\it partial} allocation that is envy-free up to the removal of any good ($\EFX$) and has NSW at least half of the optimal; here, optimality is considered across all allocations, fair or otherwise. We also prove, for subadditive valuations, the universal existence of complete allocations that are envy-free up to one good ($\EFone$) and also achieve a factor $1/2$ approximation to the optimal $\NSW$. Our $\EFone$ result resolves an open question posed by Garg, Husic, Li, V{\'{e}}gh, and Vondr{\'{a}}k (STOC 2023). 
	
	In addition, we develop a polynomial-time algorithm which, given an arbitrary allocation $\tilA$ as input, returns an $\EFone$ allocation with NSW at least $\frac{1}{e^{2/e}}\approx \frac{1}{2.08}$ times that of $\tilA$. Therefore, our results imply that the $\EFone$ criterion can be attained simultaneously with a constant-factor approximation to optimal NSW in polynomial time (with demand queries), for subadditive valuations. The previously best-known approximation factor for optimal NSW, under $\EFone$ and among $n$ agents, was $O(n)$ -- we improve this bound to $O(1)$. 
	
	It is known that $\EFone$ and exact Pareto efficiency ($\PO$) are incompatible with subadditive valuations. Complementary to this negative result, the current work shows that we regain compatibility by just considering a factor $1/2$ approximation: $\EFone$ can be achieved in conjunction with $\frac{1}{2}$-$\PO$ under subadditive valuations. As such, our results serve as a general tool that can be used as a black box to convert any efficient outcome into a fair one, with only a marginal decrease in efficiency.
}

\setcounter{page}{1}

\input{intro}

\input{prelim}

\input{poly-time-ef1}

\input{PartEfx}

\section{Fair and Efficient Allocations Under Constraints}\label{sec:constraint}
This section highlights adaptations of our results for fair and efficient allocations under constraints. Constraint fair division has been extensively studied in recent years, specifically in indivisible-items settings; see \cite{suksompong2021constraints} for a survey. Much of this literature has focused on additive valuations. 
 
Here, we provide positive results for significantly more general valuations (namely, subadditive) and under a broad class of constraints induced by downward-closed set families.  

In discrete fair division, constraints are specified by a set family $\calF\subseteq 2^{[m]}$ by requiring that each agent's bundle is contained in $\mathcal{F}$. Indeed, setting $\calF$ as the collection of independent sets in a matroid gives us matroid constraints and with $\calF$ as the collection of all connected components of a given graph we get connectivity constraints; see \cite{suksompong2021constraints} for other examples of combinatorial constraints. Formally, an allocation $\calA=(A_1,\ldots,A_n)$ (partial or complete) is said to be \textit{feasible}, under $\calF$, if $A_i\in \calF$ for each $i\in [n]$. We denote set of all feasible allocations under $\calF$ as
 $$\Pi^{\calF} := \{ (A_1,\ldots,A_n) \mid  A_i\in \calF \text{ for each } i\in [n] \}.$$ 

Here, we focus on constraints induced by downward-closed set families. Recall that a set family $\calF\subseteq 2^{[m]}$ is said to be {downward closed} if, for every subset $A\in \calF$, we have $B \in \calF$, for each $B \subseteq A$. Note that constraints imposed by matroids and, hence, cardinality constraints are downward closed. Knapsack constraints are also downward closed. 

In constraint fair division, Pareto efficiency is considered across feasible allocations. Furthermore, an allocation $\calA$ is said to be $\alpha$-PO if there is no other allocation $\calB \in \Pi^{\calF}$ such that $v_i(A_i) \leq \alpha \ v_i(B_i)$ for each $i \in [n]$. It is relevant to note that the allocation with all empty bundles is feasible under any downward closed constraint and is envy-free. However, the fact that we focus on (approximate) Pareto efficiency---in addition to fairness and feasibility---ensures that the obtained allocations provide meaningful guarantees.  

Our next result establishes that, under subadditive valuations and for \textit{any} downward-closed constraint, there exists a  feasible (partial) allocation that is both fair and approximately Pareto efficient. 

\begin{theorem}\label{thm:ConstraintExist} Every fair division instance, with subadditive valuations under a downward-closed constraint, admits a feasible (partial)  $\EFX$ allocation $ \overline{\calP}$ with the property that
	 $$\NSW(\overline{\calP}) \geq \frac{1}{2} \max_{\mathcal{X} \in \Pi^\calF} \NSW(\mathcal{X}).$$
\end{theorem}
\begin{proof} 
We obtain the theorem by invoking Algorithm~\ref{alg:Part-EFX} with feasible allocation $\calA^*=(A^*_1, \ldots, A^*_n)$ that maximizes Nash social welfare, $\calA^* \in \argmax_{\mathcal{X} \in \Pi^\calF} \NSW(\mathcal{X})$. 

We note that the partial $\EFX$ allocation $\overline{\calP}=(\overline{P}_1, \ldots, \overline{P}_n)$, guaranteed in Theorem \ref{thm:existsEFX}, additionally satisfies the following containment property: each bundle $\overline{P}_i \subseteq A^*_k$, for some $k \in [n]$. Since $A^*_k \in \calF$ and the constraints are downward closed, we get that each bundle $P^*_i \in \calF$. Hence, the partial allocation returned by Algorithm~\ref{alg:Part-EFX} is feasible, $\overline{\calP} \in \Pi^\calF$. Theorem \ref{thm:existsEFX} also gives us $\NSW(\overline{\calP}) \geq \frac{1}{2} \NSW(\calA^*)$. Since $\overline{\calP}$ is \EFX and feasible, the theorem stands proved.  
%
\end{proof}

\begin{remark}
Extending Theorem \ref{thm:ConstraintExist} to require that the feasible allocation is complete and $\EFone$ leads us to an interesting open question. In fact, existence of feasible and complete $\EFone$ allocations is open even for identical, additive valuations and under general matroid constraints; see, e.g., \cite{CES24} and references therein. 
\end{remark}

\citet{CES24} show that, for additive valuations and with matroid constraints, there exists a feasible (partial) allocation that is $\frac{1}{2}$-\EFone and $\PO$. 
Our result complements theirs by showing the existence of (partial) allocation that is \EFX and $\frac{1}{2}$-$\PO$, in the more general setting of subadditive valuations with downward-closed constraints.

We next show that, for matroid constraints, the existential guarantee of Theorem \ref{thm:ConstraintExist} has an algorithmic counterpart.
 

\begin{theorem}\label{thm:ConstrainedEF1}
Let $\mathcal{I}^{\mathcal{M}} = \langle [n], [m], \{ v_i \}_{i=1}^n , \mathcal{M}\rangle$ be a fair division instance with additive valuations and  constraints imposed by matroid $\mathcal{M}$. Then, given an $\alpha$-approximation algorithm for maximizing $\NSW$ under  submodular valuations, we can compute, in polynomial time (using the independence oracle for $\mathcal{M}$), a feasible (partial) allocation $\calA$ that is $\EFone$ and 
$$\NSW(\calA) \geq \frac{1}{\alpha e^{2/e}} \max_{\mathcal{X} \in \Pi^\calF} \NSW(\mathcal{X}).$$  
\end{theorem}
\begin{proof}
We first transform the given constraint fair division instance, $\calI^{\mathcal{M}}= \langle [n], [m], \set{v_i}_{i \in [n]} , \mathcal{M} \rangle$,  to an unconstrained instance  $\calI= \langle [n], [m], \set{\overline{v}_i}_{i \in [n]} \rangle$ with submodular valuations $\overline{v}_i$s; a similar approach was utilized in \cite{BisBar18}. In particular, for every $i \in [n]$ and each subset $S \subseteq [m]$, define valuation function $\overline{v}_i (S) \coloneqq \max \limits_{T\subseteq S : T \in \calM}v_i(T)$. Since the given valuations $v_i$s are additive, the defined ones,  $\overline{v}_i$s, are weighted matroid rank functions. Hence, $\overline{v}_i$s are (monotone) submodular \cite{schrijver2003combinatorial}. 

In addition, for any given $S \subseteq [m]$, the value $\overline{v}_i(S)$ can be computed in polynomial time (via the greedy algorithm) using the independence oracle for $\mathcal{M}$ \cite{schrijver2003combinatorial}. That is, we can answer value queries for $\overline{v}_i$s in polynomial time. 

For the constructed instance $\mathcal{I}$, we now apply the  $\alpha$-approximation algorithm for maximizing $\NSW$. Let $\tilA$ be the resulting allocation and note that  $\NSW(\tilA)\geq \frac{1}{\alpha} \NSW(\calA^*)$, where $\calA^*$ denotes the $\NSW$ optimal allocation under $\calI$. The definition of $\overline{v}_i$s implies that the optimal $\NSW$ among feasible allocations in $\calI^{\mathcal{M}}$ is also $\NSW(\calA^*)$, \ms{i.e., $\NSW(\calA^*)= \max_{\mathcal{X} \in \Pi^\calF} \NSW(\mathcal{X})$}. 
 
Next, applying \Cref{thm:PolyEF1} to instance $\calI$ and with input allocation $\tilA$, we obtain, in polynomial time, an $\EFone$ allocation $\calA=(A_1,\ldots,A_n)$ that satisfies $\NSW(\calA)\geq \frac{1}{e^{2/e}} \NSW(\tilA)\geq \frac{1}{\alpha e^{2/e}}   \NSW(\calA^*)$.
 
 We will extract the desired allocation $\calR=(R_1,\ldots,R_n)$ from the computed allocation $\calA$. Specifically, for each $i \in [n]$, define $R_i \in \argmax\limits_{T \subseteq A_i \ : \ T \in \calM} \ v_i(T)$. Here, by definition, each $R_i \in \calM$ and, hence, the allocation $\calR$ is feasible. Also, by construction, $v_i(R_i)  = \overline{v}_i (A_i)$. Therefore, $\NSW(\calR) \geq \frac{1}{\alpha e^{2/e}}   \NSW(\calA^*)$.

 
 Finally, we establish that $\calR$ is $\EFone$ in $\mathcal{I}^{\mathcal{M}}$. Recall that $\calA$ is $\EFone$ in $\calI$ (\Cref{thm:PolyEF1}). Hence, $\overline{v}_i(A_i) \geq \overline{v}_i(A_j \setminus \{g'\})$, for each $i, j \in [n]$ and some $g' \in A_j$. Further, we have $v_i(R_i)  = \overline{v}_i(A_i) \geq \overline{v}_i(A_j \setminus \{g' \})  \geq  \overline{v}_i(R_j \setminus \{g' \}) = v_i(R_j \setminus \{g'\})$. Here, the last inequality follows from the fact that $R_j \setminus \{g'\} \subseteq A_j \setminus \{g' \}$ and $\overline{v}_i$ is monotone. The last equality holds since $R_j \setminus \{g' \} \in \calM$; recall the definition of $\overline{v}_i$. Therefore, $\calR$ upholds the $\EFone$ criterion under $v_i$s, i.e., the feasible allocation $\calR$ is $\EFone$ in $\mathcal{I}^{\mathcal{M}}$. This completes the proof of the theorem. 
\end{proof}
\begin{remark} 
We can use  \Cref{thm:ConstrainedEF1} in conjunction with the $4$-approximation algorithm of  \cite{GHL23} for $\NSW$ maximization under submodular valuations.  \ms{ This instantiation shows that, under additive valuations and any matroid constraint, a feasible (partial) $\EFone$ allocation that achieves a $\frac{1}{4 e^{2/e}}$-approximation to the optimal Nash social welfare (among all feasible allocations)  can be computed in polynomial time.}
\end{remark}

\section{Conclusion and Future Work}
This work establishes that, under subadditive valuations, (partial) \EFX and \EFone properties can be attained simultaneously with a factor $\frac{1}{2}$ approximation to the optimal Nash social welfare. Our $\NSW$ approximation guarantees regarding \EFX and \EFone are both tight. These results  imply that, for subadditive valuations, (partial) \EFX and \EFone are compatible with $\frac{1}{2}$-\PO . Given that \EFone and exact \PO are incompatible for subadditive valuations, an immediate and interesting question is to determine the largest $\alpha$ for which \EFone and $\alpha$-\PO are compatible. Our results prove that 
$\alpha \geq \frac{1}{2}$.

Another interesting direction is to establish tight bounds on the compatibility between $\EFone$ and the approximations of  $\NSW$ under submodular (and gross substitutes) valuations. Under submodular valuations and with respect to $\NSW$, our results prove that the approximation factor is in the range $[\frac{1}{2},\frac{1}{e^{1/e}}]$.
 
\section*{Acknowledgments.}
Siddharth Barman gratefully acknowledge the support of the Walmart Centre for Tech Excellence (CSR WMGT-23-0001) an Ittiam CSR Grant (OD/OTHR-24-0032).
Mashbat Suzuki is supported by the ARC Laureate Project
FL200100204 on ``Trustworthy AI''.

\bibliographystyle{ACM-Reference-Format}
\bibliography{abb,bibFairDiv}

\input{appendix}

\input{etc}

\end{document}

%% file: intro.tex
\section{Introduction}
A foundational theme in fair division is the study of the interplay between fairness and efficiency; see, e.g., \citep{Moulin04} and \cite{Barbanel05} for textbook treatments, specifically for divisible resources. Indeed, optimizing the tradeoff between fairness and allocation efficiency has wide-ranging implications in many real-world domains. Motivated by such considerations, the question of when and to what extent one can achieve individual fairness in conjunction with collective welfare has been extensively studied across many disciplines, including economics, mathematics, and computer science~\citep{Moulin04,Varian73,AAB23}. The current work contributes to this theme in the context of indivisible goods. That is, we focus on settings wherein a set of indivisible goods (resources) have to be partitioned fairly and efficiently among agents who have individual preferences over the goods. 

We begin by presenting the formalisms for fairness and efficiency for the current context and then detail our results. 


\medskip

\noindent
{\bf Fairness.} A mathematical study of fairness in resource-allocation settings dates back to the work of Banach, Knaster, and \citet{Ste48}.  In the past several decades, multiple definitions of fairness have been proposed to capture requirements in various application contexts. In this body of work, envy-freeness (EF) stands out as a quintessential notion of fairness, requiring that no agent values the resources allocated to another agent more than the resources allocated to themselves.

In the setting of indivisible goods, it is impossible to guarantee envy-freeness. This nonexistence is evident even with just two agents and a single good: the agent who does not receive the good envies the one who does. As a result, meaningful relaxations of envy-freeness have been developed in the discrete fair division literature. A prominent notion here---proposed by \citet{Bud11}---is envy-freeness up to one good ($\EFone$), and it deems a partition of the indivisible goods to be fair if envy between any two agents can be resolved by removing \textit{some} good from the envied agent's bundle. 

\citet{LMM04} showed that $\EFone$ allocations of indivisible goods can be computed in polynomial time for general monotone valuations. A notable strengthening of the $\EFone$ criterion, called envy-free up to \textit{any} good (\EFX), was proposed by \citet{CKM16}. $\EFX$ requires that envy between agents can be mitigated by removing \emph{any} good from the envied agent's bundle. The existence of $\EFX$ allocations is a central open problem in discrete fair division \cite{procaccia2020technical}; the existence is not known even for instances with four agents who have additive valuations~\cite{CGM20}. Nonetheless, \cite{CKMS20} shows that by keeping a bounded number of goods unassigned---i.e., by considering partial allocations---one can obtain $\EFX$ for general monotone valuations. The main fairness notions studied in this work are $\EFone$ and (partial) $\EFX$ allocations.

\medskip

\noindent
{\bf Efficiency.} In allocative contexts, efficiency is typically measured by welfare functions \citep{Moulin04}. Indeed, allocations that optimize standard welfare functions are guaranteed to be Pareto efficient/ Pareto optimal ($\PO$). Three prevalent welfare functions in mathematical economics are utilitarian social welfare, Nash social welfare ($\NSW$), and egalitarian welfare. 
 
Here, utilitarian social welfare and egalitarian welfare are the two extreme objectives. $\NSW$ provides a balance between the two. In particular, utilitarian social welfare is defined as the sum of the agents' valuations, and egalitarian welfare is the minimum valuation across the agents. Maximizing utilitarian social welfare can lead to unfair outcomes, such as allocating all items to a single agent. On the other extreme, the egalitarian objective aims to maximize the valuation of the least satisfied agent and, hence, overlooks the utility of the collective. 

Nash social welfare is the geometric mean of the agent's valuations. It strikes a balance between the utilitarian (arithmetic mean) and the egalitarian (minimum across values) criteria. Owing to its appealing properties and axiomatic justifications, $\NSW$ stands as a prominent welfare notion in discrete fair division; see, e.g., \cite{CKM16} and \cite{AAB23}. Along these lines, the current work expresses  efficiency of allocations as their Nash social welfare.   

We also note that high $\NSW$ directly provides approximation guarantees for Pareto efficiency: for any $\alpha \in [0,1]$, an allocation that achieves $\alpha$ fraction of the optimal $\NSW$ is also $\alpha$-$\PO$.\footnote{For parameter $\alpha \in [0,1]$, an allocation $\calA$ is said to be $\alpha$-$\PO$ if there is no other allocation $\calB$ such that---for every agent $i$---$\alpha$ times $i$'s valuation in $\calB$ is at least as much as $i$'s valuation in $\calA$, and the inequality is strict for at least one agent. Here, $\alpha=1$ coincides with standard PO. Also, note that an allocation that achieves $\alpha$ fraction of the optimal $\NSW$ is also $\alpha$-$\PO$.} 


%
%

\medskip

\noindent
{\bf Fairness Along with Efficiency.}  Under additive valuations, maximizing $\NSW$ results in an \EFone outcome \cite{CKM16}. Hence, for additive valuations, one can achieve perfect compatibility between fairness and efficiency by maximizing $\NSW$. However, beyond additive valuations, maximizing Nash welfare no longer guarantees \EFone. In fact, the existence of allocations that are simultaneously \EFone and  \PO is open for submodular valuations. As for subadditive valuations, \citet{CKM16} showed that \EFone and exact \PO (and, hence, exact optimality for $\NSW$) are incompatible. 


With this backdrop and contributing to an understanding of the interplay between fairness and efficiency, the current work establishes that to regain compatibility, it suffices to consider a factor $1/2$ approximation. In particular, we prove that, for subadditive valuations, fairness (formalized by $\EFone$ and partial $\EFX$) can always be achieved with a factor $1/2$ approximation to the optimal $\NSW$; here, optimality is considered across all allocations, fair or otherwise. As mentioned, such an approximation guarantee implies that the corresponding fair allocations are $\frac{1}{2}$-$\PO$.

\subsection{Our Results}
We introduce an algorithmic approach that, despite its simplicity, generalizes existing results and leads to novel guarantees. 

Our contribution for partial $\EFX$ is stated next. 

\begin{theorem}\label{thm1:efx} Every fair division instance, with subadditive valuations, admits a partial \EFX allocation with $\NSW$ at least half of the optimal. 
\end{theorem}

\ms{This theorem generalizes the fairness and Nash welfare guarantees of \citet{CGH19} (who showed that the same bound holds for additive valuations) to the broadest class in the hierarchy of complement-free valuations, i.e., to subadditive valuations.} Also, it is relevant to note that the bound obtained in Theorem~\ref{thm1:efx} is tight, since there exists instances wherein no \EFX allocation (partial or otherwise) can achieve better than $1/2$ approximation to the optimal $\NSW$, even for additive valuations  \cite{CGH19}. 


We then extend the partial \EFX allocation obtained in Theorem~\ref{thm1:efx} to a complete \EFone allocation and obtain the following result. 
\begin{corollary}\label{thm:EF1intro} Every fair division instance, with subadditive valuations, admits an \EFone allocation with $\NSW$ at least half of the optimal. 
\end{corollary}
\Cref{thm:EF1intro} resolves an open question raised by \citet{GHL23}, where they asked,

	\begin{quote}
		\centering
	\textit{``[For submodular valuations] does there exist an $\EFone$ allocation with high NSW value?''} \\
	\end{quote}

Therefore, we answer this question affirmatively. In fact, our results are even stronger, since they hold with respect to the broader class of subadditive valuations. 

Complementing  \Cref{thm:EF1intro}, we also show that the obtained bound of $1/2$ is tight: there exist fair division instances, with subadditive valuations, wherein no $\EFone$ allocation has $\NSW$ more than half of the optimal.   

In addition, we demonstrate that the compatibility between \EFone and $\NSW$ does not hold 
in general beyond subadditive valuations. 
Specifically, we provide instances with {\it superadditive} valuations where the multiplicative gap between the optimal Nash welfare and the Nash social welfare of every \EFone allocation is unbounded.

We next focus on computational results. 
\begin{theorem}\label{thm:3poly} For subadditive valuations, there exists a polynomial-time algorithm that takes as input arbitrary allocation $\tilA$, and computes (in the value-oracle model) an \EFone allocation $\calA$ with $\NSW$ at least $\frac{1}{e^{2/e}}\approx \frac{1}{2.08}$ times that of $\tilA$. 
\end{theorem}
The above result serves as a general tool that can be used as a black box to convert any efficient allocation to a fair one, with only a marginal decrease in efficiency. To highlight this point, we note that  Theorem~\ref{thm:3poly}  can be used in conjunction with the approximation algorithm of \citet{DLRV24}, that achieves a constant-factor approximation for $\NSW$ under subadditive valuations.\footnote{The algorithm of \citet{DLRV24} requires demand-oracle access to the valuations, which is a stronger requirement than the standard value queries utilized in this work.} Hence, using demand queries for subadditive valuations, we can find, in polynomial-time, an $\EFone$ allocation that provides a constant-factor approximation for the optimal $\NSW$. Prior to this work, the best-known approximation factor for $\NSW$, under $\EFone$ and among $n$ agents, was $O(n)$ -- Theorem~\ref{thm:3poly} improves this bound to $O(1)$.

\ms{Another interesting application of our approach is in constrained fair division. In particular, we establish that, for subadditive valuations, joint fairness and efficiency guarantees hold even under arbitrary downward-closed constraints.}

\begin{theorem} 
Every fair division instance with subadditive valuations and a downward-closed constraint admits a feasible (partial) allocation that is $\EFX$ and achieves a $\NSW$ of at least half of the optimal feasible allocation.
\end{theorem}

An additional, relevant consequence of our results is the compatibility between $\EFone$ and $\frac{1}{2}$-$\PO$, under subadditive valuations. \ms{Furthermore, this compatibility extends even to settings with arbitrary downward-closed constraints.}  Prior work of \citet{CKM16} showed that $\EFone$ and (exact) $\PO$ are incompatible for subadditive valuations.\footnote{\ms{They also show that MEF1 (marginal \EFone) and $\PO$ are compatible for submodular valuations} }  
 Complementing this negative result, we establish that $\EFone$ is in fact compatible with $\frac{1}{2}$-$\PO$, under subadditive valuations. To the best of our knowledge, this is the first positive result on the compatibility of \ms{exact}  $\EFone$ and $\PO$ (albeit in an approximate sense) that goes beyond additive valuations and holds without any restricting assumptions, such as a limited number of items and agents or fixed number of marginal values of the goods. 

\subsection{Our Techniques} 

\noindent\textbf{Achieving a partial \EFX allocation with High NSW.} Our approach, which achieves partial allocation while maintaining \EFX and attaining a high $\NSW$, differs significantly from existing approaches for achieving high NSW (e.g., \citet{CGH19}, \citet{GHL23}, \citet{FMP24}). These methods typically start with the Nash optimal allocation and use a matching-based algorithm on the $\EFX$ feasibility graph to carefully remove items from the over-demanded bundles.

In contrast, we start from an empty allocation and incrementally build the allocation while maintaining EFx along with additional invariants.  The  $\EFX$ property employs the ``most envious agent'' concept introduced by \citet{CKMS20}.  Two technical insights, in particular, which we employ are: (i) In the algorithm, swaps are performed only within the optimal bundles, rather than across the set of all unallocated goods, and (ii) The algorithm maintains the invariant that, if some agent $j$ is allocated items from agent $i$'s $\NSW$ optimal bundle $A_i^*$, then $i$ is envy-free towards $j$. These technical insights play a significant role in establishing the welfare guarantee. Also, unlike the additive setting, where it is possible to ensure that each agent receives at least half the value they obtain from the Nash optimal solution, such guarantees may not be feasible in the subadditive setting. Consequently, our analysis adopts an aggregate approach. Rather than aiming for individual guarantees, we focus on lower bounding the geometric mean of values derived from the EFx partial allocation, yet another  departure from prior work. 

\medskip
\noindent\textbf{Computing \EFone allocation with High NSW.}  
To obtain our computational results for \EFone (\Cref{thm:3poly}), we consider the items as vertices of a graph, where swaps between subsets of unallocated goods are allowed only if they form a connected component in the graph.
\ms{
This ensures that each agent's partial allocation is a connected component. The connected components are then grown using carefully designed tie-breaking rules to determine which agent's bundle to update while maintaining $\EFone$.  Similar ``interval growing''  ideas have been applied for allocation of indivisible items \cite{BBN20}, and in cake cutting \cite{BarmanK23, ABK19cake}. A key difference between our approach and existing ones lies in the graph structure: we use disjoint cycles instead of a single path. This captures the insight that, instead of treating the grand bundle as a monolithic path, it is useful to consider the optimal bundles separately as cycles. This enables us to obtain structural guarantees and develop novel counting arguments to achieve our approximation results.
 }

\medskip
\noindent\textbf{Outline.} Theorem \ref{thm:3poly} (restated formally as Theorem~\ref{thm:PolyEF1}) is established in Section \ref{section:ef1-algo}. We prove the existential results---in particular, Theorem \ref{thm1:efx} and \Cref{thm:EF1intro} (restated as Theorems \ref{thm:existsEFX} and \ref{thm:existEF1}, respectively)---in Section \ref{section:efx-exist}. In \Cref{sec:constraint}, we explore the implications of our results for the problem of fair and efficient allocation with constraints.   


\subsection{Additional Related Work}
\textit{Axiomatic Properties of $\NSW$.} The fairness and efficiency properties of Nash social welfare are well-studied in economics \citep{Moulin04,Varian73,KN79}. In particular, $\NSW$ satisfies the Pigou-Dalton principle. Another notable property of $\NSW$ is that it scale free, i.e., if an agent scales her valuation multiplicatively, independent of others, it does not change the order in which all the allocations are ranked via $\NSW$~\cite{Moulin04}. \citet{Suks23} showed that, for additive valuations, $\NSW$ is the only welfarist rule which satisfies \EFone. 

\medskip

\noindent
\textit{Approximation Algorithms for $\NSW$ maximization.} Maximizing Nash social welfare is {\rm NP}-hard even when  agents have identical, additive valuations. 
Given its significance in discrete fair division, considerable effort has been dedicated to developing approximation algorithms for $\NSW$ maximization  
\citep{CG15,CDGJMVY17,LiV21b,LiV21a,GargHV21}. For additive valuations, the best known approximation guarantee is $e^{1/e} \approx 1.44$ \cite{BKV18}.  
For submodular valuations, the best known approximation ratio is $4$ \cite{GHL23}. Recently, in a notable result, \citet{DLRV24} gave a constant factor ($\approx 375000$) approximation algorithm for $\NSW$ maximization under subadditive valuations. 

\medskip

\noindent
\textit{\EFX Allocations.} As mentioned previously, the existence of complete allocations (i.e., allocations in which all the goods have been assigned among the agents) that are $\EFX$ is a central open problem in discrete fair division~\cite{procaccia2020technical}. \citet{PlautR20} showed the existence of \EFX allocations among agents with identical valuations. They also established the existence of $\frac{1}{2}$-\EFX allocations under subadditive valuations. \citet{CKMS20} showed that, among $n$ agents, partial \EFX allocations exist under general monotone valuations, with an upper bound of  $n-1$ on the number of unallocated goods.  Subsequently, the bound on the number of unallocated goods was improved to $n-2$ \citep{BCFF22,Mahara24}. 

\medskip

\noindent
\textit{Fair and Efficient Allocations.} For additive valuations, the existence of fair and efficient allocation---in terms of \EFone and \PO---was established in \cite{CKM16}. \citet{BKV18} showed that, for additive valuations, \EFone and \PO allocations can be computed in pseudo-polynomial time. 


For subadditive valuations, \citet{CGM21} gave an $O(n)$ approximation to maximum Nash welfare that satisfies \EFone and either of the guarantees: $\frac{1}{2}$-\EFX or partial \EFX with at most  $n-1$  unallocated goods. In a recent work, for subadditive valuations, \citet{GHL23} showed the existence of $\frac{1}{2}$-\EFX allocations that approximate the optimal $\NSW$ to within a factor of $\frac{1}{2}$. Subsequently, \citet{FMP24} improved this bound by proving the existence of $\frac{1}{2}$-\EFX allocations with $\NSW$ at least $\frac{2}{3}$ times the optimal. \\

\noindent\textit{Price of Fairness.}  The price of fairness quantifies the loss of collective welfare incurred to ensure fairness \cite{bertsimas2011price,CKK09}. \citet{BLM19} initiated the study of the price of fairness in the context of indivisible item allocation, focusing on utilitarian social welfare. When the fairness notion studied is \EFone, \citet{BBN20} established a tight bounds on the price of fairness for utilitarian social welfare.  \ms{For \EFX,  \cite{bu2025approximability} provides bounds on the price of fairness for utilitarian social welfare and also study the complexity questions regarding utilitarian welfare maximization under fairness constraints. 
	Note that our results can be interpreted as showing that the price of fairness for $\EFone$ with respect to $\NSW$ is exactly $1/2$, for subadditive valuations.
  }

%% file: prelim.tex
\section{Notation and Preliminaries}
We study fair division of $m \in \mathbb{Z}_+$ indivisible goods among $n \in \mathbb{Z}_+$ agents with individual preferences over the goods. We will denote the set of all goods and the set of all agents as $[m] =\{1, 2, \ldots, m\}$ and $[n] = \{1, 2, \ldots, n\}$, respectively. The cardinal preferences of the agents $i \in [n]$ over the goods are captured by valuations (set functions) $v_i: 2^{[m]} \mapsto \mathbb{R}_+$. In particular, $v_i(S) \in \mathbb{R}_+$ denotes the (nonnegative) value that agent $i \in [n]$ has for any subset of goods $S \subseteq [m]$. We will denote a fair division instance by the tuple $\langle [n], [m], \{v_i\}_{i=1}^n \rangle$. Also, for ease of notation, to denote the value that agent $i \in [n]$ has for a good $g \in [m]$, we will write $v_i(g)$ (instead of $v_i(\{g\})$). 

Addressing division of goods (non-negatively valued resources), the current work assumes that agents $i \in [n]$ valuations are (i) nonnegative: $v_i(S) \geq 0$, for all $S \subseteq [m]$, (ii) normalized: $v_i(\emptyset) = 0$, and (iii) monotone: $v_i(X) \leq v_i(Y)$, for all $X \subseteq Y \subseteq [m]$. 

A set function $v$ is said to be subadditive if it satisfies $v(S \cup T) \leq v(S) + v(T)$, for all subsets $S, T \subseteq [m]$. Subadditive valuations model complement-freeness and constitute the broadest function class in the well-studied hierarchy of complement-free valuations. Recall that subadditive functions include additive, submodular, and XOS valuations as special cases. \ms{Consequently, the results in this work apply broadly by establishing the compatibility between fairness and efficiency for all the mentioned complement-free valuation classes.} 


An allocation $\calA = (A_1, \ldots, A_n)$ refers to an $n$-partition of the set of goods $[m]$. Here, the subset of goods $A_i \subseteq [m]$ is assigned to agent $i \in [n]$ and is called agent $i$'s bundle. Note that, in an allocation $\calA = (A_1, \ldots, A_n)$, the bundles are pairwise disjoint ($A_i \cap A_j = \emptyset$ for all $i \neq j$) and satisfy $\cup_{i=1}^n A_i = [m]$. We will use the term partial allocation $\calP = (P_1, \ldots, P_n)$ to denote an assignment of goods among the agents, with pairwise disjoint bundles $P_i \subseteq [m]$ that might not contain all the $m$ goods, i.e., $\cup_{i=1}^n P_i \subsetneq [m]$. In addition, we will use the term complete allocation to highlight divisions wherein all the goods have been necessarily assigned.

The notions of fairness and welfare studied in this work are defined next. 
\begin{definition}[Envy-free up to one good]
In a fair division instance $\langle [n], [m], \{v_i\}_{i=1}^n \rangle$, an allocation $\calA = (A_1, \ldots, A_n)$ is said to be envy-free up to one good ($\EFone$) if for every pair of agents $i, j \in [n]$ (with $A_j \neq \emptyset$), there exists a good $g \in A_j$ such that $v_i(A_i) \geq v_i(A_j \setminus \{g\})$. 
\end{definition}

A strengthening of the $\EFone$ criterion is obtained by demanding that envy between agents can be mitigated by the removal of \emph{any} good from the envied agent's bundle. The existence of $\EFX$ allocations, even under additive valuations, stands as a major open problem in discrete fair division. Multiple, recent works have also considered multiplicative approximations of this fairness notion. Formally, 

\begin{definition}[Envy-Free Up to Any Good ($\EFX$)]
For fair division instance $\langle [n], [m], \{v_i\}_{i=1}^n \rangle$ and parameter $\alpha \in [0,1]$, an allocation $\calA=(A_1, \ldots, A_n)$ is said to be an $\alpha$-$\EFX$ allocation if for all pairs of agents $i, j \in [n]$ (with $A_j \neq \emptyset$) and every good $g \in A_j$ we have $v_i(A_i) \geq \alpha \ v_i(A_j \setminus \{g\})$. 
\end{definition}
Here, $\alpha =1$ corresponds to exact $\EFX$ allocations. 

In a fair division instance $\langle [n], [m], \{v_i\}_{i=1}^n \rangle$, the Nash social welfare ($\NSW$) for an allocation $\calA = (A_1, \ldots, A_n)$---denoted as $\NSW(\calA)$---is defined as the geometric mean of the values that the agents obtain under the allocation 
$\NSW(\calA) \coloneqq \left( \prod_{i=1}^n v_i(A_i) \right)^{1/n}$.
In any given instance, we will write $\calA^*=(A^*_1, \ldots, A^*_n)$ to denote an allocation with maximum possible $\NSW$ among all allocations.


Our algorithms operate in the standard value oracle model. That is, in terms of an input specification, the algorithms only require an oracle, which when given any agent $i \in [n]$ and subset $S \subseteq [m]$, returns the value $v_i(S) \in \mathbb{R}_+$.

%% file: poly-time-ef1.tex
\section{Finding Fair Allocations with High Nash Social Welfare}
\label{section:ef1-algo}   
This section provides a polynomial-time algorithm (Algorithm \ref{alg:PolyIGC}) which, given an arbitrary allocation $\tilA$ as input, returns an $\EFone$ allocation with $\NSW$ at least $\frac{1}{e^{2/e}} \approx \frac{1}{2.08}$ times that of $\tilA$.  

In Algorithm \ref{alg:PolyIGC}, we will index the goods considering the given allocation $\tilA=(\widetilde{A}_1, \ldots ,\widetilde{A}_n)$. Specifically,  in each bundle $\widetilde{A}_i$, let the goods be indexed as $g^i_1, g^i_2, \ldots, g^i_{|\widetilde{A}_i|}$. Further, we construct a graph $\mathcal{G}$ containing (disjoint) $n$ cycles, one for each bundle $\widetilde{A}_i$; see Figure~\ref{fig:G}. In particular, the $i$th cycle is over $|\widetilde{A}_i|$ vertices that correspond the goods $g^i_j$ respectively. Note that any connected component $U$---in the cycle for $\widetilde{A}_i$---can be identified as $\{g^i_{x}, g^i_{x+1}, \ldots, g^i_{y} \}$, where the indices $x,\ldots,y$ are considered modulo $|\widetilde{A}_i|$.

\begin{figure}[h]
	\centering
	\includegraphics[width=0.7\textwidth]{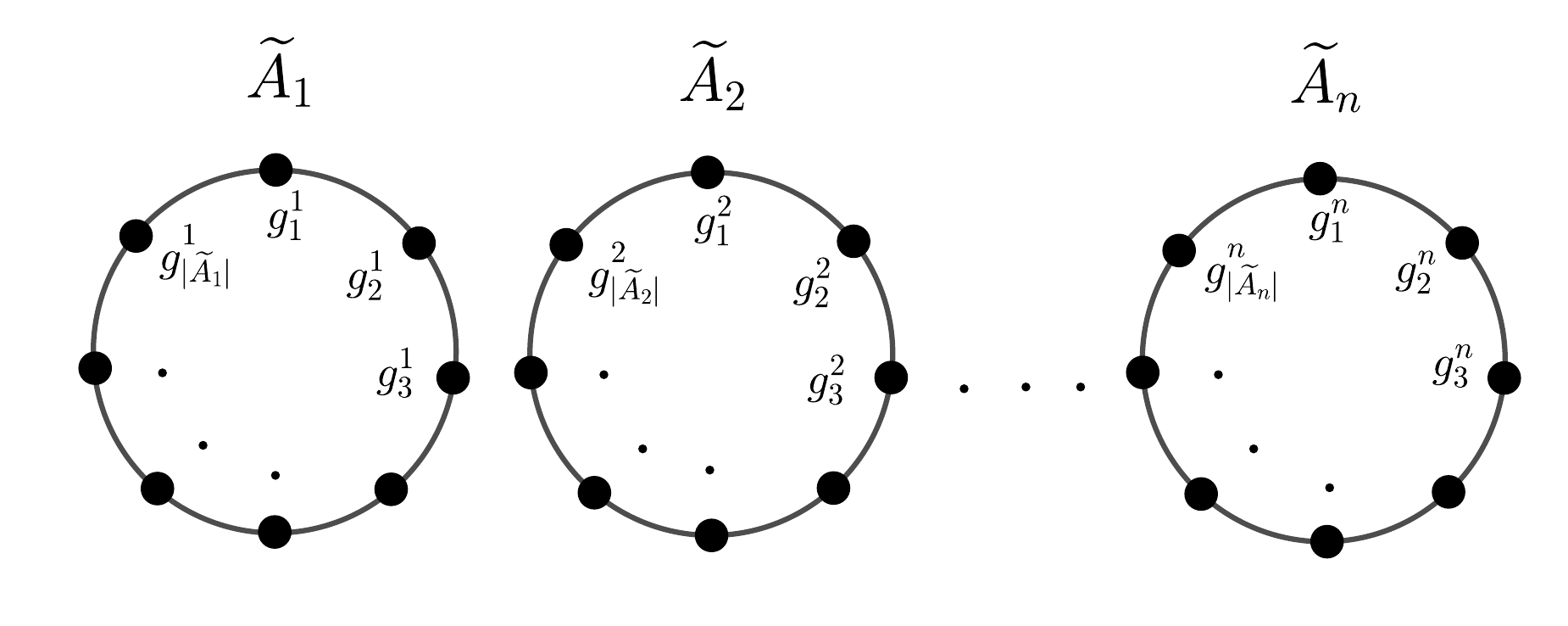}
		\caption{The graph $\mathcal{G}$ corresponding to allocation $\tildA$}\label{fig:G}
\end{figure}

{\centering
	\begin{minipage}[t]{\linewidth}
		\begin{algorithm}[H]
			\caption{Polynomial-Time Path Growing}
			\label{alg:PolyIGC}
			\DontPrintSemicolon
			
			\KwIn{ A fair division instance $\calI= \langle [n], [m], \set{v_i}_{i \in [n]} \rangle$ with value oracle access to the subadditive valuations $v_i$'s.  Allocation $\tilA=(\widetilde{A}_1,...,\widetilde{A}_n)$. }
			\KwOut{ \EFone allocation with high NSW.}
			
			For each agent $i\in [n]$,  initialize $P_i=\emptyset $. 
			
			Considering $\tilA$, construct  graph $\mathcal{G}$ ($n$ disjoint cycles) as detailed above.  
			
			
			Let $\mathcal{U}(\calP)$ denote the collection of connected components (paths or cycles) in $\mathcal{G}$ that remain after the removal of $\bigcup_{i\in [n]} \ P_i$.
			
			\While{there exists an agent $i\in [n]$ and a path or cycle $U \in \mathcal{U}(\calP)$ such that  $v_i(P_i)< v_i(U)$ }
			{ 
				Let $1 \leq s \leq |U|$ be the smallest index such that $E(s, U, \calP )\neq \emptyset$; the while-loop condition ensures that such an index $s$ exists. 
				
				Also, let ${j} \in [n]$ be the index such that connected component $U$ is contained in the cycle for bundle $\widetilde{A}_j$, i.e., $U = \{g^j_a, \ldots, g^j_b\}$ with $a,b \in \{1, \ldots, |\widetilde{A}_j|\}$.
				
				
					
					\uIf{ agent  $j \in E(s, U, \calP ) $}
					{Update $P_j = \{ g^j_{a}, \ldots, g^j_{a+s-1} \}$, while keeping partial assignment of all other agents unchanged.
					}
					\Else
					{Select any agent $k\in E(s, U, \calP )$ and update $P_k= \{ g^j_{a}, \ldots, g^j_{a+s-1} \}$, while keeping partial assignment of all other agents unchanged.
					}
					

				%
				%

				Update $\mathcal{U}(\calP)$ to be the collection of connected components in the graph $\mathcal{G}$ that remain after the removal of $\bigcup_{i\in [n]} \ P_i$.
			}
			
			Extend the partial ($\EFone$) allocation $\calP=(P_1,...,P_n)$ to a complete allocation $\calA=(A_1,...,A_n)$ using the envy-cycle elimination algorithm of \citet{LMM04} (this extension ensures that $v_i(A_i) \geq v_i(P_i)$ for all $i$). 
			
			\Return Allocation $\calA$ 
		\end{algorithm}
	\end{minipage}
}


Also, for any set of goods $U = \{g^i_x, g^i_{x+1}, \ldots, g^i_{y-1}, g^i_y\}$ with consecutive indices\footnote{As mentioned, the indices are considered modulo $|\widetilde{A}_i|$} (i.e., for any connected component in $\mathcal{G}$) and an index $1 \leq s \leq |U|$, we define $E(s, U, \calP)$ to be the subset of agents that envy the first $s$ goods in $U$ under partial allocation $\calP=(P_1, \ldots, P_n)$; formally,  envying agents $E(s, U, \calP) \coloneqq \big\{ r \in [n] \mid v_r(P_r) < v_r(\{g_x, \ldots, g_{x+s-1} \})  \big\}$. Note that this set can be empty. 

Algorithm \ref{alg:PolyIGC} proceeds by iteratively updating bundles of the agents with envied sets. The envied sets are selected judiciously to ensure that the $\EFone$ criterion is maintained throughout. Moreover, the algorithm in its while-loop always assigns bundles that induce connected components (paths or cycles) in the graph $\mathcal{G}$. At a high level, this design decision (i.e., the focus on bundles that form paths or cycles in $\mathcal{G}$) ensures that at most polynomially many bundles can be considered for allocation. This construction essentially leads to polynomial-time termination of the algorithm. The properties of the allocation computed by Algorithm \ref{alg:PolyIGC} are stated in Theorem \ref{thm:PolyEF1}. 

\begin{theorem}\label{thm:PolyEF1} 
	Let $\mathcal{I} = \langle [n], [m], \{ v_i \}_{i=1}^n \rangle$ be a fair division instance in which the agents' valuations $v_i$s are subadditive. Given value oracle access to $v_i$s and an allocation $\tilA$, we can compute, in polynomial time, a complete $\EFone$ allocation $\calA$ with the property that 
	$$\NSW(\calA)\geq \frac{1}{e^{2/e}}\NSW(\tilA).$$
\end{theorem}

\begin{proof}
	Let  $\overline{\calP}=( \overline{P}_1, \ldots, \overline{P}_n)$ be the partial allocation obtained in Algorithm~\ref{alg:PolyIGC} when its while-loop terminates. Note that in each iteration of the while-loop the valuation of the selected agent, $j$ or $k$, strictly increases and the valuations of the remaining agents remain unchanged.  This observation implies that the loop necessarily terminates, i.e., the partial allocation $\overline{\calP}$ is guaranteed to exist. We will show that $\overline{\calP}$ is $\EFone$ and upholds the stated $\NSW$ approximation guarantee. 
	
	Recall that the envy-cycle elimination algorithm of \citet{LMM04} extends a partial $\EFone$ allocation to a complete $\EFone$ allocation without decreasing the valuation of any agent. Hence, the $\EFone$ and Nash welfare guarantees (established below) of partial allocation $\overline{\calP}$ continue to hold for the returned, complete allocation $\calA=(A_1, \ldots, A_n)$.
	
	The partial allocation $\overline{\calP}$ is $\EFone$: In each iteration of the while-loop, the minimality of the selected index $s$ ensures that, for each agent $i \in [n]$ and against the selected agent $\ell \in \{j, k\}$, envy-freeness holds up to the removal of largest-index good $g^j_{a+s-1}$ from $P_\ell = \{g^j_a, \ldots, g^j_{a+s-1} \}$. Hence, inductively, we get that the partial allocation $\overline{\calP}$ is $\EFone$. 
	
	We now focus on the $\NSW$ approximation guarantee for $\overline{\calP}=( \overline{P}_1, \ldots, \overline{P}_n)$. First, note that by the termination condition of the while loop, we have $v_i(\overline{P}_i) \geq v_i(U)$ for all $i\in [n]$ and for each $U\in \mathcal{U}(\overline{\calP})$. 
	
	By construction, each $\overline{P}_r$ and every $U \in \mathcal{U}(\overline{\calP})$ correspond to a connected component in exactly one of the $n$ cycles that constitute the graph $\mathcal{G}$. That is, each $\overline{P}_r$ and $U \in \mathcal{U}(\overline{\calP})$ is contained in exactly one of the bundles among $\widetilde{A}_1, \ldots, \widetilde{A}_n$.  With this observation, for each agent $i \in [n]$, write $S_i$ and $L_i$ to denote the collection of bundles in $\overline{\calP}$ and  $\mathcal{U}(\overline{\calP})$, respectively, that are contained in $\tildA_i$. Specifically, write $L_i \coloneqq \{  U \in \mathcal{U}(\overline{\calP}) \mid U \subseteq \tildA_i  \}$.  Further, let $S_i \coloneqq \{ r \in [n] \mid \overline{P}_r \subseteq \tildA_i  \}$\footnote{\ms{ We can assume, without loss of generality, that $\tildA_i\neq \emptyset$ for all $i\in [n]$, and hence $\overline{P}_r$s are non-empty. } }.  Note that the sets $S_i$'s form a partition of $[n]$.  That is, $S_i \cap S_j =\emptyset$, for all $i \neq j$, and  $\bigcup_{i\in [n]} S_i= [n]$. 
	
	Lemma~\ref{lem:S_i} (stated and proved below) shows that, for each $i \in [n]$ and $r \in S_i$,  we have $v_i\left(\overline{P}_r \right) \leq v_i(\overline{P}_i)$. 
	
	Next, we note that $\tildA_i=\left(\bigcup\limits_{U\in L_i} U  \right) \cup \left( \bigcup\limits_{r\in S_i} \overline{P}_r \right)$ and invoke the subadditivity of the valuation $v_i$ to obtain  
	\begin{align} 
		v_i(\tildA_i) & \leq   \sum_{U\in L_i} v_i(U  ) + \sum_{r\in S_i} v_i\left(\overline{P}_r  \right) \nonumber \\
		& \leq  |L_i| v_i(\overline{P}_i)  + \sum_{r\in S_i} v_i\left(\overline{P}_r  \right)  \tag{since $ v_i(U)\leq  v_i(\overline{P}_i)$} \\ 
		&\leq  |L_i| v_i(\overline{P}_i)+ |S_i| v_i(\overline{P}_i)  \tag{Lemma~\ref{lem:S_i} } \\ 
		&= (|L_i|+|S_i|)v_i(\overline{P}_i) \label{eq: optpoly1}
	\end{align}

	We will extend inequality (\ref{eq: optpoly1}) using the following claim.  
	
	\medskip
	\begin{claim}
		\label{claim:intersection-count}
		Across agents $i \in [n]$, with nonempty $S_i$, the sum of the counts $|S_i| + |L_i|$ is upper bounded as follows: $\sum\limits_{i\in [n] \ : \ |S_i|>0 \  } \left( |S_i| +|L_i| \right) \leq 2n$. Furthermore, 
		\begin{align*}
		\left( \prod_{i \in [n]: \ |S_i| >0 } (|L_i|+|S_i|) \right)^{\frac{1}{n}} \leq e^{2/e}.  
		\end{align*}
	\end{claim}
	\begin{proof}[Proof of Claim]  
		By construction, each of the $n$ bundles $\overline{P}_i$ induces a connected component in the graph $\mathcal{G}$. Further, for any $i\in [n]$ with $|S_i|>0$, we have that $|L_i|\leq |S_i|$; this inequality follows from the observation that after removing $|S_i|>0$ connected components from a cycle at most $|S_i|$ connected components remain.  Therefore, 		
		\begin{align} 
			\sum_{i\in [n]: \ |S_i|>0} \left(|S_i| +|L_i| \right)  \leq \sum_{i\in [n]:  \ |S_i|>0} 2 |S_i| =    2n \label{ineq:intersection-count}
		\end{align} 
			Here, the last  equality follows from the fact that $S_i$'s form a partition of $[n]$, i.e., $\sum_{i\in [n]}|S_i| =n$
			
Next, write $\ell$ to denote the number of agents with a nonempty $S_i$, i.e., $ \ell \coloneqq | \{i \in [n] : |S_i| > 0 \}|$. For the geometric mean of the $\ell$ counts $(|S_i| + |L_i|)$, we have 
\begin{align*}
\left( \prod_{i \in [n]: \ |S_i| >0 } (|S_i| + |L_i| ) \right)^\frac{1}{\ell} \underset{\text{AM-GM ineq.}}{\leq} \left( \frac{1}{\ell} \sum_{i : \ |S_i| > 0}  (|S_i| + |L_i| ) \right) \underset{\text{via (\ref{ineq:intersection-count})}}{\leq} \frac{2n}{\ell}. 
\end{align*}
Therefore, 
\begin{align*}
\left( \prod_{i \in [n]: \ |S_i| >0 } (|L_i|+|S_i|) \right)^{\frac{1}{n}} \leq  \left( \frac{2 n}{\ell} \right)^{\frac{\ell}{n}} \leq e^{2/e}. 
\end{align*}
The last inequality is obtained by observing that, in the domain $x \in (0,1]$, the maximum value of the function $f(x) \coloneqq \left( \frac{2}{x} \right)^{x}$ is equal to $e^{2/e}$; the maximum is achieved at $x = \frac{2}{e}$. Hence, $\left( \frac{2 n}{\ell} \right)^{\frac{\ell}{n}} \leq f \left( \frac{2}{e} \right)  = e^{2/e}$. This completes the proof of the claim.  
\end{proof}

It is relevant to note that for any agent $z \in [n]$, if $S_{z}$ is empty, i.e., $|S_{z}| =0$, then the entire cycle (in the graph $\mathcal{G}$) for the bundle $\widetilde{A}_{z}$ is contained in $\mathcal{U}(\overline{\calP})$. Hence, for any such agent $z$, with $|S_z| =0$, the termination condition of the while-loop gives us 
\begin{align}
v_{z}(\widetilde{A}_{z}) \leq v_{z} ( \overline{P}_{z}) \label{ineq:zero-agents}
\end{align} 
	
	To establish the desired bound on the NSW of $\overline{\calP}$, we note that 
	\begin{align}
		\left( \prod_{i\in [n]} v_i(\tildA_i) \right)^{\frac{1}{n}} &= \left( \prod_{i\in [n] : \ |S_i|>0} v_i(\tildA_i) \right)^{\frac{1}{n}} \left( \prod_{i\in [n] : \ |S_i|=0} v_i(\tildA_i) \right)^{\frac{1}{n}} \nonumber \\ 
		&\leq  \left( \prod_{i\in [n] : \ |S_i|>0} v_i(\tildA_i) \right)^{\frac{1}{n}} \left( \prod_{i\in [n] : \ |S_i|=0} v_i(\overline{P_i}) \right)^{\frac{1}{n}} \tag{via (\ref{ineq:zero-agents})} \\  
		 &\leq  	\left( \prod_{i\in [n]: \ |S_i|>0} (|L_i|+|S_i|)  \ v_i( \overline{ P_i}) \right)^{\frac{1}{n}}  \left( \prod_{i\in [n] : \ |S_i|=0} v_i(\overline{P_i}) \right)^{\frac{1}{n}} \tag{via (\ref{eq: optpoly1})}\\ 
		&= \left( \prod_{i\in [n]: \ |S_i|>0} (|L_i|+|S_i|)\right)^{\frac{1}{n}}   \left( \prod_{i\in [n]} v_i(\overline{ P_i}) \right)^{\frac{1}{n}}  \nonumber \\ 
		&\leq e^{2/e}  \NSW(\overline{\calP}) \tag{Claim \ref{claim:intersection-count}}
	\end{align}

Therefore, the partial $\EFone$ allocation $\overline{\calP}$ satisfies the stated NSW bound. 
	
	To complete the proof of the theorem it remains to bound the time complexity of Algorithm~\ref{alg:PolyIGC}. 
	
	\medskip
	
	\noindent\textbf{Runtime Analysis.} 
	Recall that the algorithm of \citet{LMM04} runs in polynomial time. Hence, to establish the computational efficiency of Algorithm~\ref{alg:PolyIGC}, it suffices to show that the while-loop in the algorithm terminates in polynomial time.  
	
	We note that in the graph $\mathcal{G}$ (defined as a disjoint union of $n$ cycles over $m$ vertices/goods) there are at most $O(n m^2)$ connected components. Also, throughout the execution of the while-loop, the assigned bundles $P_i$ correspond to one of the $O(n m^2)$ connected components. Furthermore, as the loop iterates, the valuations of the agents increase monotonically. These observations imply that, for any agent $i \in [n]$, the assigned bundle $P_i$ can be updated at most  $O(n m^2)$ times. Since each iteration of the while-loop updates the bundle of one agent, the loop must terminate after $O(n ^2 m^2)$  iterations. Observing that each iteration of the loop executes in polynomial time, we obtain that the while-loop finds the partial allocation $\overline{\calP}$ in polynomial time. As mentioned above, the computational efficiency of the while-loop implies that Algorithm~\ref{alg:PolyIGC} overall runs in polynomial time.
	
	In summary, Algorithm~\ref{alg:PolyIGC} efficiently finds an $\EFone$ allocation $\calA$ with high NSW. The theorem stands proved. 
\end{proof}

\begin{lemma}\label{lem:S_i} Let  $\overline{\calP}=( \overline{P}_1, \ldots, \overline{P}_n)$ be the partial allocation obtained in Algorithm~\ref{alg:PolyIGC} when its while-loop terminates. Also, let $S_i \coloneqq \{ r \in [n] \mid \overline{P}_r \subseteq \tildA_i \}$. Then,   
	\begin{align}\label{ineq: envy}
		v_i (\overline{P_i}) \geq v_i (\overline{P_r} ) \qquad \text{for each } i\in [n] \text{ and }r \in S_i 
	\end{align}
\end{lemma}
\begin{proof} 
	We show, by a contradiction, that the stated bound holds in every iteration of the while-loop. Since the initial partial allocation consists of empty bundles, equation (\ref{ineq: envy}) holds trivially. Now, consider the iteration of the while-loop wherein equation (\ref{ineq: envy}) is violated for the first time. Denote $\calP'$ and $\calP''$ as the partial allocations maintained at the beginning and the end of the iteration, respectively. Analogously, we denote  $S'_i \coloneqq \{ r \in [n] \mid P'_r \subseteq \tildA_i  \}$ and $S''_i \coloneqq \{ r \in [n] \mid P''_r \subseteq \tildA_i\}$, for each $i\in [n]$ . We have that $v_i (P'_i) \geq v_i (P'_r)$  for each  $i\in [n]$  and  $r \in S'_i $. Let $\ell \in [n]$ be the agent whose bundle got updated during the iteration. Note that $\ell$ is either the identified agent $j$ or selected as $k$. Also, let $X$ be the new bundle she receives, i.e., $P''_\ell  = X$. Since in each iteration of the while-loop only one agent's bundle changes, we have that $P''_i = P'_i$ for all agents $ i \neq \ell$. It then follows that $S''_i \subseteq S' _i \cup \{ \ell \} $ for each $i\in [n]$. Now, observe that, for the selected agent $\ell$ and each $r \in S''_\ell \setminus \{ \ell \}$, we have 
	\begin{align*}
		v_\ell(P''_\ell) > v_\ell(P'_\ell) & \geq v_\ell(P' _r)  =  v_\ell(P'' _r).
	\end{align*}
Note that the inequality $v_\ell(P''_\ell) \geq  v_\ell(P'' _r)$ extends to all $r \in S''_\ell$, since this bound holds trivially if $r = \ell$.

	In addition, $v_i(P''_i) = v_i(P'_i) \geq v_i(P'_r)  =  v_i(P'' _r)$ for each $i \neq \ell$ and $r \in S''_i \setminus \{ \ell \}$. Hence, for $\calP''$ to violate equation (\ref{ineq: envy}),  there must exist an agent $j \in [n]\setminus \{ \ell \}$ such that $v_j(P''_j)= v_j(P'_j)< v_j(P''_\ell)$ (and the selected agent $\ell \in S''_j$). It follows that, for  index $j$ (see Line~6), $P''_\ell \subseteq U \subseteq \tildA_j$, here, $U$ is the connected component identified in the while-loop condition. Further, with $s$ as the index identified in Line~5, we have $P''_\ell = \{ g^j_{a}, \ldots, g^j_{a+s-1} \}$ and, hence, $j, \ell \in E(s,  U, \calP')$. However, by the \textsc{If}-condition in Line~7, we get that $j$ must have been the agent whose bundle got updated, as opposed to agent $\ell=k$. This contradicts the assumption that $\ell$'s bundle got updated in the current iteration. Therefore, in each iteration of the while-loop equation (\ref{ineq: envy}) is satisfied. Hence, the allocation $\overline{\calP}$ obtained at the termination of the while-loop also satisfies equation (\ref{ineq: envy}). This completes the proof of the lemma. 
\end{proof}

\begin{remark} A useful consequence of Theorem \ref{thm:PolyEF1} is that, given any $\alpha$-approximation algorithm for $\NSW$, we can compute, in polynomial time, an $\EFone$ allocation that approximates $\NSW$ within a factor of $\frac{1}{\alpha e^{2/e}} $.

\smallskip

At first glance, the construction utilized in Theorem \ref{thm:PolyEF1} may appear simplistic since the algorithm finds local improvements (whenever possible) while maintaining fairness as an invariant within a seemingly arbitrarily imposed graph structure. 
However, it is relevant to note that the approximation guarantee obtained in Theorem \ref{thm:PolyEF1} is quite close to the best possible: Theorem \ref{theorem:ef1-lowerb} (Section \ref{section:ef1-lowerb}) shows that, even in terms of an existential guarantee, one cannot always expect $\EFone$ allocations with $\NSW$ more than $1/2$ times  the optimal.
\end{remark}

%% file: PartEfx.tex
\section{Existence of High-Welfare $\EFX$ Allocations under Subadditive Valuations}
\label{section:efx-exist}

This section establishes, for subadditive valuations, the universal existence of partial allocations that are both $\EFX$ and have $\NSW$ at least half of the optimal. Our proof here is constructive: we develop an algorithm (Algorithm \ref{alg:Part-EFX}) which, given a Nash optimal allocation $\calA^*$, goes on to find a partial $\EFX$ allocation $\overline{\calP}$ with only a marginal decrease in Nash social welfare.

{\centering
	\begin{minipage}[t]{\linewidth}
		\begin{algorithm}[H]
			\caption{$\EFX$ Set Growing}
			\label{alg:Part-EFX}
			\DontPrintSemicolon
			
			\KwIn{A fair division instance $\calI= \langle [n], [m], \set{v_i}_{i \in [n]} \rangle$ with value oracle access to the subadditive valuations $v_i$'s. A Nash optimal allocation $\calA^*=(A^*_1,...,A^*_n)$.}
			\KwOut{Partial $\EFX$ allocation with high NSW.}
			
			For each $i \in [n]$, initialize the assigned bundle $P_i = \emptyset$ and the set of unassigned goods $U_i = A^*_i$.
			
			\While{there exist subset $U_j \subseteq [m]$ and agent $i\in [n]$ such that  $v_i(P_i)< v_i(U_j)$ } 
			{ 
				 Let $X \subseteq U_j$ be an inclusion-wise minimal set with the property that $v_k(P_k) < v_k(X)$, for some $k \in [n]$. 
				
				\eIf{$v_j(P_j) < v_j(X)$}
				{
					Update $P_j=X$, while keeping partial assignment of all other agents unchanged.
				}
				{	Select an arbitrary agent $k \in [n]$ for whom $v_k(P_k) < v_k(X)$. 
					
					Update $P_k=X$, while keeping partial assignment of all other agents unchanged. 
				} 	
				Update $U_i = A_i^*\setminus \left( \bigcup_{r \in [n]} P_r \right)$, for each $i\in [n]$.  
			} 
			
			
				
				
			
			
			
			\Return Partial allocation $\calP=(P_1, \ldots, P_n)$. 
		\end{algorithm}
\end{minipage}
}

\begin{theorem}\label{thm:existsEFX} Every fair division instance, with subadditive valuations, admits a partial $\EFX$ allocation $\overline{\calP}$ with Nash social welfare $\NSW(\overline{\calP}) \geq \frac{1}{2}  \NSW(\mathcal{A}^*).$ 	
	Here, $\calA^*$ denotes the Nash optimal allocation in the given instance. 
	\end{theorem}
\begin{proof}
We will show that the partial allocation returned by Algorithm \ref{alg:Part-EFX} upholds the stated fairness and welfare guarantees. Write $\overline{\calP}=(\overline{P}_1, \ldots, \overline{P}_n)$ to denote the partial allocation obtained by Algorithm \ref{alg:Part-EFX} when the while-loop terminates, i.e., the algorithm returns $\overline{\calP}$. 

In each iteration of the while-loop in Algorithm \ref{alg:Part-EFX}, for the (envied) set $X$ and the selected agent $k \in [n]$, with bundle $P_k$, we have $v_k(P_k) < v_k(X)$. Hence, in any iteration of the loop, assigning $X$ (in lieu of $P_k$) to the selected agent strictly increases her valuation. At the same time, the valuations of all the other agents remain unchanged. That is, throughout the execution of the algorithm the sum of the agents' valuations increases monotonically. This observation implies that the loop terminates and, hence, the partial allocation $\overline{\calP}$ is guaranteed to exist. 

We show below that the Nash social welfare of the returned partial allocation $\overline{\calP}$ is at least $1/2$ times the optimal.

\medskip
\noindent\textbf{Partial allocation $\overline{\calP}$ satisfies $\NSW(\overline{\calP})\geq \frac{1}{2}\NSW(\calA^*)$.}  The construction of the while-loop in Algorithm~\ref{alg:Part-EFX} ensures that, for each agent $r \in [n]$, the bundle $\overline{P}_r \subseteq U_i$, for some agent $i \in [n]$. Also, by construction, $U_i \subseteq A^*_i$ for each $i \in [n]$. Therefore, for every agent $r \in [n]$, the bundle $\overline{P}_r \subseteq A^*_i$, for some agent $i \in [n]$.

We define, for each $i\in [n]$, the set $S_i \coloneqq \{r \in [n] \mid \ \overline{P}_r \subseteq A^*_i \}$.\footnote{\ms{We can assume, without loss of generality, that $A^*_i\neq \emptyset$ for all $i\in [n]$, and hence $\overline{P}_r$s are non-empty. } } Note that the sets $S_i$'s form a partition of $[n]$. That is,  $S_i \cap S_j =\emptyset$, for all $i \neq j$, and  $\bigcup_{i\in [n]} S_i= [n]$. 

We will next show that, for every agent $i \in [n]$, the inequality $v_i(\overline{P}_i) \geq v_i( \overline{P}_r)$ holds for each $r \in S_i$. In particular, we will show, via induction, that these bounds hold for each partial allocation $\calP$ obtained in the while-loop. Here, the base case of the induction holds since the initial partial allocation consists of empty bundles. 

Now, for the induction step, fix an iteration of the while-loop. Let $\calP'=(P'_1, \ldots, P'_n)$ and $\calP''=(P''_1, \ldots, P''_n)$ denote the partial allocations maintained at the beginning and end of the iteration, respectively. Also, let $a \in [n]$ be the agent whose bundle gets updated in the iteration. Note that $a =j$, when the if-condition in Line 4 of the algorithm holds. Otherwise, $a = k$, as selected in Line 7 of the algorithm. 

For the selected agent $a$ we have $P''_a = X$ and for all other agents $i \neq a$, the bundle remains unchanged in the iteration, $P''_i = P'_i$. We will use $S'_i$ and $S''_i$ to capture the containments of the current partial allocations, $S'_i \coloneqq \{r \in [n] \mid \ P'_r \subseteq A^*_i \}$ and $S''_i \coloneqq \{r \in [n] \mid \ P''_r \subseteq A^*_i \}$.

Note that for the set $U_j$ considered in the iteration, we have the containment $P''_a \subseteq U_j$, which in turn implies $S''_j = S'_j \cup \{ a \}$. Besides $a$, the bundles assigned to all agents remain unchanged in the iteration, hence, for all $i \neq j$, we have $S''_i \subseteq S'_i$. Furthermore, the induction hypothesis ensures that $v_i(P'_i) \geq v_i(P'_r)$, for every $i \in [n]$ and each $r \in S'_i$. These observations ensure that, for every agent $i \neq j$ and each $r \in S''_i \subseteq S'_i$ we have the required inequality: 
\begin{align*}
v_i(P''_i) & \geq v_i(P'_i)  \\
& \geq v_i(P'_r) \tag{induction hypothesis} \\
& = v_i(P''_r) \tag{$r \in S''_i$ and, hence, $r \neq a$}
\end{align*}

An analogous bound holds for agent $j$ and each $r \in S''_j$. Specifically, the inequality $v_j(P''_j) \geq v_j(P''_r)$ follows from the induction hypothesis for all $r \in S''_j \setminus \{a \}$. Now, for agent $a \in S''_j$, we consider two cases:

\noindent 
{\it Case {\rm I}: $a = j$.} Here, the inequality $v_j(P''_j) \geq v_j(P''_a)$ directly holds. 

\noindent
{\it Case {\rm II}: $a \neq j$.} In this case, the else-block (Lines 7 and 8) must have executed in the considered iteration of the while-loop. Equivalently, the if-criterion did not hold in the iteration. That is, in the current case, we have $v_j(P'_j) \geq v_j(X) = v_j(P''_a)$. Also, $P''_j = P'_j$, since $a \neq j$. Therefore, in this case as well, we have $v_j(P''_j) \geq v_j(P''_a)$.

Overall, the required inequality  $v_i({P}''_i) \geq v_i( {P}''_r)$ holds every $i \in [n]$ and each $r \in S''_i$. This completes the induction and implies that even the allocation $\overline{\calP}$ obtained at the termination of the while-loop satisfies 
\begin{align}
v_i(\overline{P}_i) \geq v_i(\overline{P}_r) \qquad \text{for all $i \in [n]$ and each $r \in S_i$} \label{ineq:look-within}
\end{align}

Next, note that the termination condition of the loop implies that $v_i(\overline{P}_i)\geq v_i(U_j)$, for all agents $i\in [n]$ and $j\in [n]$. Also, we have
 $A_i^*=U_i\cup \bigcup_{r\in S_i} \overline{P}_r$. Therefore,   
	 \begin{align} 
	 	v_i(A_i^*) &\leq v_i(U_i) + \sum_{r\in S_i} v_i(\overline{P}_r) \tag{$v_i$ is subadditive} \nonumber \\
	 	&\leq  v_i(\overline{P}_i)+  |S_i| v_i(\overline{P}_i) \tag{via inequality (\ref{ineq:look-within})} \nonumber \\ 
	 	&= (1+|S_i|) \ v_i(\overline{P}_i) \label{eq: opt}
	 	\end{align}
	Finally, using inequality (\ref{eq: opt}), we obtain  
\begin{align*}
	\left( \prod_{i\in [n]} v_i(A^*_i) \right)^{\frac{1}{n}} &\leq  	\left( \prod_{i\in [n]} (1+ |S_i|)v_i(\overline{P}_i) \right)^{\frac{1}{n}} \\ 
	&= \left( \prod_{i\in [n]} (1+ |S_i|)\right)^{\frac{1}{n}}   \left( \prod_{i\in [n]} v_i(\overline{P}_i) \right)^{\frac{1}{n}}   \\ 
	&\leq \left(\frac{1}{n} \sum_{i\in [n]} (1+ |S_i|)  \right)\NSW(\overline{\calP}) \tag{AM-GM inequality} \\ 
	&= 2 \NSW(\overline{\calP}).
	\end{align*}
	Here, the last  inequality follows from the fact that $S_i$'s form a partition of $[n]$, i.e., $\sum_{i\in [n]}|S_i| =n$.
	
This establishes the stated $\NSW$ guarantee for $\overline{\calP}$. 

\medskip
	
\noindent\textbf{Partial allocation $\overline{\calP}$ is $\EFX$.} The fairness of the partial allocation $\overline{\calP}$ follows from fact that in each iteration of the while-loop we select an inclusion-wise minimal (envied) set $X$. 


Specifically, we will show that the while-loop maintains $\EFX$ as an invariant. Towards an inductive argument, note that the initial allocation---with $P_i = \emptyset$ for all $i \in [n]$---is trivially $\EFX$. 

Fix any iteration of the loop. Write $\calP'=(P'_1, \ldots, P'_n)$ to denote the partial allocation maintained at the beginning of the considered iteration and $\calP''=(P''_1, \ldots, P''_n)$ be the partial allocation at the end. Also, let $a \in [n]$ be the agent whose bundle is updated in the iteration. Note that $P''_i = P'_i$, for all agents $i \neq a$, and $P''_a = X$. By the induction hypothesis, we have that $\calP'$ is $\EFX$. 

The minimality of the set $X$ identified in the iteration ensures that for every strict subset $Y \subsetneq X$ we have $v_i(P'_i) \geq v_i(Y)$ for all agents $i$. Equivalently, for every good $g \in X$, we have $v_i(P'_i) \geq v_i( X \setminus \{g\})$. Given the update performed in the iteration, we note that the partial allocation $\calP''$ is also $\EFX$: 
\begin{itemize}
\item For all agents $i \neq a $ and against agent $a$, we have $v_i(P''_i) = v_i(P'_i) \geq v_i( X \setminus \{g\}) = v_i(P''_a \setminus \{g\})$ for all goods $g \in P''_a = X$.
\item For agent $a$ and against all other agents $i$, we have $v_a(P''_a) = v_a(X) > v_a(P'_a) \underset{\text{ind.~hyp.}}{\geq} v_a(P'_i \setminus \{g \}) = v_a(P''_i \setminus \{g \})$ for all goods $g \in P''_i = P'_i$.
\item The $\EFX$ criterion holds for all other pairs of agents $i, j \neq a$ via the induction hypothesis. 
\end{itemize}
This completes the induction and shows that  $\overline{\calP}$ is an $\EFX$ allocation. \\
	
Overall, we have that the returned partial allocation $\overline{\calP}$ is $\EFX$ and satisfies $\NSW(\overline{\calP})\geq \frac{1}{2}\NSW(\calA^*)$. This completes the proof. 
\end{proof}

\subsection{Complete $\EFone$ and $\frac{1}{2}$-$\EFX$ Allocations with High Nash Social Welfare}
In this section, we show that under subadditive valuations, there exist  a high-welfare, complete allocations that are $\EFone$. In fact, the allocations are also guaranteed to be $\frac{1}{2}$-$\EFX$. This result is achieved by building on Algorithm \ref{alg:Part-EFX}.

Specifically, we start with the partial allocation returned by Algorithm~\ref{alg:Part-EFX} and assign all the unallocated goods using the envy-cycle elimination algorithm of \citet{LMM04}; see Algorithm \ref{alg:EFone}. The properties of envy-cycle elimination directly imply that the resulting complete allocation is $\EFone$ and continues to have high Nash social welfare. 

Further, using the properties of the allocation returned of Algorithm~\ref{alg:Part-EFX} and with a tailored analysis of the envy-cycle elimination, one can show that allocation returned by Algorithm \ref{alg:EFone} is $\frac{1}{2}$-$\EFX$ as well. We remark that parts of the $\frac{1}{2}$-$\EFX$ analysis follow the arguments used in \cite{PlautR20} and \cite{GHL23}.

{\centering
	\begin{minipage}[t]{\linewidth}
		\begin{algorithm}[H]
			\caption{Complete Set Growing}
			\label{alg:EFone}
			\DontPrintSemicolon
			
			\KwIn{A fair division instance $\calI= \langle [n], [m], \set{v_i}_{i \in [n]} \rangle$ with value oracle access to the subadditive valuations $v_i$'s. A Nash optimal allocation $\calA^*=(A^*_1,...,A^*_n)$.}
			\KwOut{Complete \EFone ($\frac{1}{2}$-$\EFX$) allocation $\calA=(A_1,...,A_n)$ with high NSW.}
			
		Compute partial allocation $\overline{\calP}$ by executing Algorithm \ref{alg:Part-EFX} with given instance $\calI$ and Nash optimal allocation $\calA^*$.
						
			Extend the partial allocation $\overline{\calP}$ to a complete allocation $\calA=(A_1,...,A_n)$ using the envy-cycle elimination algorithm of \citet{LMM04} (this extension ensures that $v_i(A_i) \geq v_i(\overline{P}_i)$ for all $i$). 
						
			\Return Allocation $\calA$ 
		\end{algorithm}
	\end{minipage}
}

\begin{theorem}\label{thm:existEF1} Every fair division instance, with subadditive valuations, admits an $\EFone$ (complete) allocation $\mathcal{A}$ (that is also $\frac{1}{2}$-$\EFX$) with Nash welfare $\NSW(\mathcal{A})\geq \frac{1}{2}  \NSW(\mathcal{A}^*)$. Here, $\calA^*$ denotes the Nash optimal allocation in the given instance. 
\end{theorem}
\begin{proof}
Let $\mathcal{A}$ be the allocation returned by Algorithm \ref{alg:EFone}. The envy-cycle elimination algorithm of \citet{LMM04} extends a partial allocation to a complete  allocation without decreasing the valuation of any agent. Thus, the Nash welfare guarantee of partial allocation $\overline{\calP}$ (Theorem \ref{thm:existsEFX}) continues to hold for the returned allocation $\calA$. Hence,  we have $\NSW(\mathcal{A})\geq \frac{1}{2}  \NSW(\mathcal{A}^*)$. We now show that the allocation $\mathcal{A}$ satisfies the stated fairness properties of  $\EFone$ and $\frac{1}{2}$-$\EFX$.

The partial allocation $\overline{\calP}$ computed by Algorithm \ref{alg:Part-EFX} is $\EFX$ (Theorem \ref{thm:existsEFX}) and, hence, it is $\EFone$. The envy-cycle elimination algorithm of \citet{LMM04} extends a partial $\EFone$ allocation to a complete $\EFone$ allocation thus indeed $\mathcal{A}$ is \EFone.

We now show that $\mathcal{A}$ is $\frac{1}{2}$-$\EFX$. Note that the partial allocation $\overline{\calP}$ is computed in Algorithm \ref{alg:EFone} by invoking Algorithm \ref{alg:Part-EFX} as a subroutine. Considering the termination condition of the while-loop in Algorithm \ref{alg:Part-EFX}, we obtain $v_i \left( \overline{P}_i \right) \geq v_i \left( {U}_j \right)$, for each $i, j \in [n]$; here, ${U}_j$ is the set of unassigned goods within $A^*_j$, i.e., $U_j = A^*_j \setminus \left( \bigcup_{r \in [n]} \overline{P}_r \right)$. Further, the fact that the valuations are monotone gives us $v_i(\overline{P}_i) \geq v_i(U_j) \geq v_i(g)$ for all $g \in U_j$. Therefore, each agent $i \in [n]$ values her bundle, $\overline{P}_i$, at least as much as any unassigned good $g \in \bigcup_{j=1}^n U_j = [m] \setminus \left( \bigcup_{r \in [n]} \overline{P}_r \right)$, i.e., 
\begin{align}
v_i(\overline{P}_i) \geq v_i(g) \quad \text{for all $g \in [m] \setminus \left( \bigcup_{r \in [n]} \overline{P}_r \right)$} \label{ineq:good-wise}
\end{align} 

Next, we will show, inductively, that the $\frac{1}{2}$-$\EFX$ guarantee holds in every iteration of the envy-cycle elimination method. Here, the base case holds since $\overline{\calP}$ is $\EFX$ (hence, also $\frac{1}{2}$-$\EFX$). For the induction step, fix any iteration of the envy-cycle elimination method and consider the good $g$ assigned to an unenvied agent $a$ in the iteration. Also, let $\calP = (P_1, \ldots, P_n)$ be the maintained partial allocation in the envy-cycle elimination method before this update and $\calP'=(P'_1, \ldots, P'_n)$ be the one after the update. Hence, $P'_a = P_a \cup \{g \}$ and  $P'_i = P_i$ for all $i \neq a$. The partial allocation $\calP$ is $\frac{1}{2}$-$\EFX$, via the induction hypothesis. Also, note that $g$ was unallocated before the considered iteration and, hence, $g \in [m] \setminus \left( \bigcup_{r \in [n]} \overline{P}_r \right)$. Hence, for any good $h \in P'_a$ and each agent $i \neq a$, we have 
	\begin{align*}
		v_i(P'_a \setminus \{h \}) & \leq v_i \left(P'_a \right)  \tag{monotonicity of $v_i$} \\
		&  \leq v_i \left(P_a \right) +  v_i \left(g  \right)  \tag{$P'_a = P_a \cup \{g\}$ and $v_i$ is subadditive} \\
		&  \leq   v_i \left( P_i \right)  + v_i \left(g  \right)  \tag{$a$ is an unenvied agent} \\
		& \leq 2 v_i \left( P_i \right) \tag{via (\ref{ineq:good-wise})} \\
		& = 2 v_i(P'_i) \tag{$P'_i = P_i$ for all $i \neq a$}
	\end{align*}
This bound shows that the $\frac{1}{2}$-$\EFX$ criterion holds against agent $a$ in the partial allocation $\calP'$. In addition, note that $v_a(P'_a) \geq v_a (P_a)$ and the bundles of all other agents $i \neq a$ remain unchanged between $\calP$ and $\calP'$. Hence, for all pairs of agents $i \in [n]$ and $j \in [n] \setminus \{a\}$, the partial allocation $\calP'$ inherits the $\frac{1}{2}$-$\EFX$ property from $\calP$. Therefore, inductively, we have that the complete allocation $\calA$ obtained at the end of the envy-cycle elimination method is $\frac{1}{2}$-$\EFX$. 

Overall, we get that the returned allocation $\calA$ is $\EFone$, $\frac{1}{2}$-$\EFX$, and satisfies the stated welfare guarantee. This completes the proof of the theorem. 
\end{proof}

%

\subsection{Tight Lower Bound for $\EFone$}
\label{section:ef1-lowerb}
We now show that the approximation guarantee for $\EFone$ (obtained in Theorem \ref{thm:existEF1}) is tight: For subadditive valuations, one cannot always expect $\EFone$ allocations with $\NSW$ more than half of the optimal.  

\begin{theorem}
\label{theorem:ef1-lowerb}
	For any $\varepsilon>0$, there exist fair division instances, with subadditive  valuations, wherein no $\EFone$ allocation has Nash social welfare more than $\left(\frac{1}{2} +\varepsilon \right)$ times the optimal, i.e., for \textit{every} $\EFone$ allocation $\calA$ we have
	$$\NSW(\mathcal{A}) \leq \left( \frac{1}{2} +\varepsilon \right) \NSW(\mathcal{A}^*).$$ 
	Here, $\calA^*$ denotes the Nash optimal allocation in the given instance.
\end{theorem}
\begin{proof}
	Consider a fair division instance with $m=2n-1$ indivisible goods and $n$ agents. The $m$ goods are partitioned into $n$ sets: $G_1=\{g_1\}$ and  $G_i=\{ g_{i}^1,g_{i}^2\}$ for each $2\leq i \leq n$.  Each agent $1\leq i \leq n-1$ has non-zero marginal value only for goods in $G_i \cup G_{i+1} $, and agent $n$ has non-zero marginal values only for goods in $G_n$ . With fixed parameter $\delta \in \left(0, \frac{1}{10} \right)$, we set the valuations of agents $1\leq i \leq n-1$ as follows: 
	\begin{align*}
		v_i(S) \coloneqq 
		\begin{cases}
			(1-\delta)|S \cap G_i|  \  &\text{if} \quad |S \cap G_{i+1}|=0  \\
			\max \left\{ 1,(1-\delta)|S \cap G_i|  \right\} \quad  &\text{if} \quad |S \cap G_{i+1}|=1  \\
			\max \left\{ 1+\delta,(1-\delta)|S \cap G_i|  \right\} \quad &\text{if} \quad |S \cap G_{i+1}|=2  \\
		\end{cases}
	\end{align*}
	For agent $n$, the valuation $v_n(S) \coloneqq (1-\delta)|S \cap G_n|$. Note that the agents' valuations are subadditive. Here, for each agent $i$,  it suffices to check subadditivity only for subsets of $G_i\cup G_{i+1}$.  
	
	In this instance, the $\NSW$ of an optimal allocation $\calA^*$ is at least as much as the $\NSW$ obtained by assigning $G_i$ to agent $i$ thus the optimal NSW is satisfies $\NSW(\calA^*) \geq \left( (1-\delta)^{n} 2^{n-1} \right)^{1/n} = (1-\delta) 2^{1-\frac{1}{n}}$. 
	
	Next, we bound the Nash welfare of any \EFone allocation. In particular, we will show that in any \EFone allocation $\calA$, every agent satisfies $v_i(A_i)\leq (1+\delta)$. 
	
	Fix any $\EFone$ allocation $\calA=(A_1, \ldots, A_n)$ and assume, towards a contradiction, that there exists an agent $j$ who gets a bundle of value more than $(1+\delta)$, i.e., $ v_j(A_j) \geq  (1-\delta)2$.\footnote{Note that, by the construction of the valuations, if for any agent $i$ and subset $S \subseteq [m]$ we have $v_i(S)> (1+\delta)$, then  $v_i(S) \geq  (1-\delta)2 $.} Note that as agent $1$'s valuation is at most $(1+\delta)$,  we have $j\geq 2$. Further, it must be the case that  $|A_j\cap G_{j}|\geq 2$ -- this is a necessary condition for agent $j$ to obtain value $(1-\delta) 2$. 
	
	Note that $|G_{j}|=2$ and $|A_j\cap G_{j}|\geq 2$. Hence, for agent $j-1$, we have $|A_{j-1}\cap G_{j}|=0$. Now, for agent $j-1$ to be \EFone towards agent $j$, it must hold that $G_{j-1}\subseteq A_{j-1}$. Otherwise, agent $j-1$ would get at most value $1-\delta$ and would value agent $j$'s bundle at least $1$ after removal of any good. Applying this argument iteratively we obtain $|A_2\cap G_{2}|\geq 2$. However, this inequality implies that agent $1$  \EFone envies agent $2$,  contradicting the assumption that the underlying allocation $\calA$ is $\EFone$. 
	
	Therefore, by way of contradiction, we obtain that in any \EFone allocation $\calA$, for each agent $i$ the valuation $v_i(A_i)\leq (1+\delta)$. Hence, $\NSW(\calA)\leq (1+\delta)$. Finally, the above-mentioned bounds, lead to   
	\begin{align}
		\frac{\NSW(\calA)}{\NSW(\calA^*)} & \leq \frac{(1+\delta)}{(1-\delta) 2^{1-\frac{1}{n}}} \label{eq:hardinstance}
	\end{align} 
	For any $\varepsilon >0$, a sufficiently small $\delta >0$ and a sufficiently large $n \in \mathbb{Z}_+$ ensure that the of right-hand-side of equation~(\ref{eq:hardinstance}) is at most $\frac{1}{2} + \varepsilon$. Therefore, for any $\varepsilon>0$, there exists an instance wherein $\frac{\NSW(\calA)}{\NSW(\calA^*)} \leq \frac{1}{2}+\varepsilon$. The theorem stands proved. 
\end{proof}

\subsection{Asymmetric Nash Social Welfare}
Our compatibility results generalize to settings wherein efficiency is quantified via asymmetric Nash social welfare. Formally, given an instance $\langle [n], [m], \{v_i\}_{i=1}^n \rangle$ and a tuple $\vec{w}=(w_1,\ldots,w_n) \in \mathbb{R}^n_+$, with $\sum_{i=1}^n w_i = 1$, the asymmetric Nash social welfare of an allocation $\calA=(A_1,\ldots,A_n)$ is defined as $\NSW_{\vec{w}} (\calA) \coloneqq \prod_{i=1}^n \left( v_i(A_i) \right)^{w_i}$. Asymmetric $\NSW$ aims to capture fair division among agents with varying entitlements, $w_i$'s. 
Also, note that when all the entitlements are equal---$w_i = \frac{1}{n}$ for all $i$---we get back the standard $\NSW$.  

The result below establishes a compatibility between $\EFone$ and asymmetric $\NSW$. Analogous results hold when considering fairness in terms of partial $\EFX$ and complete $\frac{1}{2}$-$\EFX$.  \ms{We primarily focus on $\EFone$  instead of weighted $\EFone$ ($\textsc{WEF1}$) here because, beyond additive valuations (even with two agents and submodular valuations), weighted \EFone is not guaranteed to exist \cite{chakraborty2021weighted}.}
	
\begin{corollary}	
For every fair division instance $\langle [n], [m], \{v_i\}_{i=1}^n \rangle$, with subadditive valuations, and entitlement $\vec{w}=(w_1,...,w_n) \in \mathbb{R}^n_+$, there exists a complete $\EFone$ allocation $\mathcal{A}$ which achieves  $( nw_{\max}+1)$-approximation to the optimal asymmetric $\NSW$. Here, $w_{\max} \coloneqq \max_{i \in [n]} \ w_i$. 
\end{corollary}
\begin{proof}  
Let  $\calB^*=(B_1^*,\ldots,B_n^*)$ denote an allocation that maximizes the asymmetric $\NSW$, with entitlements $\vec{w}=(w_1,...,w_n)$. We show that Algorithm~\ref{alg:EFone}, when executed with $\calB^*$ as input, returns an $\EFone$ allocation with high asymmetric $\NSW$.  

Let $\calA$ denote the allocation returned by Algorithm~\ref{alg:EFone}. The same argument as in the proof of Theorem~\ref{thm:existEF1} shows that $\calA$ is $\EFone$. We now focus on the $\NSW$ approximation guarantee. Note that Algorithm \ref{alg:Part-EFX} is called within Algorithm~\ref{alg:EFone}. Write $\overline{\calP}=(\overline{P}_1, \ldots, \overline{P}_n)$ be the partial allocation obtained via the the call to Algorithm \ref{alg:Part-EFX}. Using arguments similar to the ones used in Theorem~\ref{thm:existEF1} (see inequality (\ref{eq: opt})), for each agent $i \in [n]$, we have $v_i(B_i^*) \leq  (1+ |S_i|) \ v_i( \overline{P}_i)$. 
		Therefore, we can upper bound the optimal $\NSW$ in the asymmetric case as follows  
		\begin{align*}
			\prod_{i\in [n]} \left(v_i(B^*_i)\right)^{w_i} &\leq  	 \prod_{i\in [n]} \left( (1+ |S_i|) \ v_i(\overline{P}_i) \right)^{w_i} \\ 
			&=  \prod_{i\in [n]} \left(1+ |S_i|\right)^{w_i}    \prod_{i\in [n]} \left( v_i(\overline{P}_i) \right)^{w_i}   \\ 
			&\leq \left( \sum_{i\in [n]} w_i (1+ |S_i|)  \right)\NSW_{\vec{w}}(\overline{\calP}) \tag{weighted AM-GM inequality} \\ 
			&= \left( 1+ \sum_{i\in [n]} w_i |S_i|  \right)\NSW_{\vec{w}}(\overline{\calP})  \tag{since $\sum_{i\in [n]} w_i =1$} \\ 
			&\leq ( n w_{\max} +1) \NSW_{\vec{w}}(\overline{\calP}) \tag{since $\sum_{i\in N}|S_i| =n$}
		\end{align*}		
Throughout the execution of the envy-cycle elimination method (Line 2 in Algorithm \ref{alg:EFone}), the agents' valuations do not decrease. Hence, for the returned allocation $\calA$ and each agent $i \in [n]$, we have $v_i(A_i) \geq v_i(\overline{P}_i)$. This inequality and bound on the asymmetric $\NSW$ of $\overline{\calP}$ lead to the stated guarantee: $( n w_{\max} +1) \NSW_{\vec{w}}({\calA}) \geq \prod_{i\in [n]} \left(v_i(B^*_i)\right)^{w_i}$. 	The corollary stands proved. 
		\end{proof}

%% file: appendix.tex
\newpage
\appendix

\section{Appendix}

\subsection{Envy-Cycle Elimination Algorithm of \citet{LMM04} }
For completeness, in this section we provide a description of the envy-cycle elimination algorithm from \citet{LMM04}.
For any  (partial) allocation $\calA$, the  algorithm maintains a directed graph, $G_{\calA} = ([n], E)$, which captures the envy between the agents under $\calA$. Specifically, the $n$ vertices in the graph correspond to $n$ agents, respectively. Further, there is a directed arc $(i,j)\in E$ if and only if $i$ envies $j$'s bundle, i.e.,  $v_i(A_i)<v_i(A_j)$.

{\centering
	\begin{minipage}[t]{\linewidth}
		\begin{algorithm}[H]
			\caption{Envy-Cycle Elimination}
			\label{alg:ECE}
			\DontPrintSemicolon
			
			\KwIn{A fair division instance $\calI= \langle [n], [m], \set{v_i}_{i \in [n]} \rangle$ with value oracle access to the monotone valuations $v_i$'s. Partial allocation $\calP=(P_1,...,P_n)$ that is $\EFone$.}
			\KwOut{Complete \EFone allocation $\calA=(A_1,...,A_n)$.}
			
			Initialize allocation $\calA=(A_1,...,A_n) $ by setting $A_i = P_i$ for each $i\in [n]$.
			
			Let $M = [m]\setminus \left( \bigcup_{i\in [n]} P_i \right)$ be the unallocated goods. 
				
			\While{$M\neq \emptyset $} 
			{\label{line:outer} 
			\While{There exists a directed cycle in $G_{\calA}$}
			{\label{line:inner}
				Eliminate the cycle by reassigning the bundles in the reverse direction of the cycle.
			}
			
			Select an arbitrary source (unenvied) node $i \in [n]$ in $G_{\calA}$ and any unallocated good $g \in M$. 
			
			Update   $A_i \leftarrow A_i \cup \{ g \} $.
			
			Update $M\leftarrow M\setminus \{ g \}$. 
		}
			\Return{Allocation $\calA$}  
		\end{algorithm}
\end{minipage}
}

The algorithm maintains the $\EFone$ property as an invariant in each iteration of the outer while-loop (Line~\ref{line:outer}). This is because  eliminating cycles in the envy-graph preserves   $\EFone$. Furthermore, 
 if an item $g$ is allocated to a source node $i$ in $G_{\calA}$, then each agent $j\in [n]\setminus \{ i \}$  is $\EFone$ towards agent $i$, since we have $v_j(A_j)\geq v_j\left( \left( A_i \cup \{ g \} \right) \ \setminus \{ g \} \right)=v_j(A_i)  $. This  inequality holds since $i$ was a source (i.e., unenvied) node in the envy-graph. Together with the fact that the initial allocation $\calP$ is $\EFone$, these observations show that the algorithm maintains the $\EFone$ criterion in each iteration of the while-loop (Line~\ref{line:outer}). Therefore, the returned allocation is $\EFone$ as well. 

In addition, note that  the returned complete allocation $\calA$ satisfies $v_i(A_i)\geq v_i(P_i)$, for each $i\in [n]$, since the agents' values only weakly increase during the execution of the algorithm . 

Let us now analyze the running time of the algorithm. Note that, for each fixed (partial) allocation $\calA$ encountered during the algorithm,  finding and eliminating a cycle in the envy-graph $G_{\calA}$ takes $O(n^2)$ time. Further, whenever bundles are reallocated along an envy-cycle then at least one edge is removed from the envy-graph. Hence, the inner while-loop (Line~\ref{line:inner}) can be executed in time at most $O(n^4)$. Since the number of unallocated items decreases by one in each iteration of the outer while-loop (Line~\ref{line:outer}), we conclude that the algorithm runs in polynomial time. Overall, we obtain the following theorem. 

\begin{theorem}[\citet{LMM04}]
\label{theorem:lmm}
Let $\mathcal{I} = \langle [n], [m], \set{v_i}_{i \in [n]} \rangle$ be a fair division instance with monotone valuations $v_i$s. Then, given value oracle access to the $v_i$s and any partial $\EFone$ allocation $\mathcal{P} = (P_1, \ldots, P_n)$ as input, the envy-cycle elimination algorithm (Algorithm \ref{alg:ECE}) finds, in polynomial time, a complete $\EFone$ allocation $\mathcal{A}=(A_1, \ldots, A_n)$ with the property that $v_i(A_i) \geq v_i(P_i)$, for all agents $i \in [n]$. 
\end{theorem}

Indeed, instantiating Theorem \ref{theorem:lmm} with $P_i = \emptyset$, for all $i \in [n]$, yields an efficient algorithm for finding $\EFone$ allocations under monotone valuations. 

\subsection{Stronger Guarantee for Two Agents}
Here, we focus on the special case of two agents ($n=2$) and show that an $\NSW$ guarantee, stronger than the one obtained in Theorem \ref{thm:existEF1}, can be obtained for $\EFone$ allocations in this case.

\begin{theorem}\label{thm:twoagent} In any fair division instance among two agents ($n=2$) with subadditive valuations, there always exists an $\EFone$ allocation $\calA$ with the property that    
	$
	\NSW(\mathcal{A}) \geq \frac{1}{\sqrt{2}}    \NSW(\mathcal{A}^*)
	$. Here, $\calA^*=(A^*_1, A^*_2)$ denotes the Nash optimal allocation in the given instance. 
\end{theorem}
\begin{proof} 
	If the Nash optimal allocation $\calA^*$ is envy-free, then the theorem holds trivially. Next, note that $\calA^*$ is Pareto efficient and, hence, no envy-cycle exists under $\calA^*$. That is, one of the agents does not envy the other. Hence, without loss of generality, we may assume that agent 1 does not envy agent 2, but agent 2 envies agent 1, i.e., $v_1(A^*_1)\geq v_1(A^*_2)$ and $v_2(A^*_2) < v_2(A^*_1)$. Let $Z\subseteq A_1^*$ be an inclusion-wise minimal subset with the property that $v_2 (Z)> v_2(A_2^*)$. We divide the analysis into the following two cases. 
	
	\noindent		
	{\it Case 1: $v_1(Z)\geq \frac{1}{2}v_1(A_1^*)$.} 	Consider the partial allocation $(Z, A_2^*)$. In this partial allocation, if there is an envy cycle, then after exchanging the bundles, we obtain a partial  envy-free allocation with NSW is at least $\frac{1}{\sqrt{2}} \NSW(\calA^*)$. On the other hand, if there is no envy cycle under $(Z, A_2^*)$, then only agent 2 envies agent 1, since $v_2(Z)> v_2(A_2^*)$. However, the minimality of $Z$ ensures that the envy of agent 2 towards 1 can be mitigated by the removal of one good. Hence, in this case, $(Z, A_2^*)$ is an $\EFone$ allocation with NSW is at least $\frac{1}{\sqrt{2}} \NSW(\calA^*)$. \\

	\noindent	
	{\it Case 2: $v_1(Z)< \frac{1}{2}v_1(A_1^*)$.}  Consider the partial allocation $(A_1^*\setminus Z,  Z)$ obtained by partitioning $A^*_1$ among the two agents. In the current case, the subadditivity of $v_1$ gives us $v_1(A_1^*\setminus Z)> \frac{1}{2} v_1(A_1^*)$. Hence, in the partial allocation $(A_1^*\setminus Z,  Z)$, agent 1 does not envy agent 2, $v_1(A_1^*\setminus Z)> \frac{1}{2}v_1(A_1^*)>v_1(Z)$. If the $\EFone$ criterion holds for agent 2 towards agent 1, then we have a partial $\EFone$ allocation, $(A_1^*\setminus Z,  Z)$, with NSW is at least $\frac{1}{\sqrt{2}}$ times the optimal. Hence, we may assume that agent 2 envies agent 1 by more than one good. In such a scenario, we iteratively move goods from agent 1 to agent 2 until one of the following two conditions is satisfied for the first time:
	\begin{itemize}
		\item Agent 1 envies agent 2. 
		\item $\EFone$ criterion holds for agent 2 towards agent 1.
	\end{itemize}
	If agent 1 envies agent 2, then there is an envy cycle. Here, by permuting the bundles we obtain an allocation wherein agent 1 gets a bundle whose value is at least $\frac{1}{2}v_1(A_1^*)$ and agent 2 gets a bundle whose value is at least $v_2(Z)> v_2(A_2^*)$. If the second condition occurs, then we have $\EFone$ from agent 2 towards agent 1 and agent 1 does not envy agent 2. Here again, we have a partial $\EFone$ allocation with the desired approximation guarantee. \\ 
	
	In both the cases above, we have shown the existence of a partial $\EFone$ allocation with Nash welfare at least $1/\sqrt{2}$ times the optimal. Using the algorithm of \citet{LMM04}, we can extend such a partial $\EFone$ allocation to a complete $\EFone$ allocation, while maintaining the NSW approximation guarantee. This establishes the existence of the desired $\EFone$ allocation and completes the proof of the theorem. 		 
\end{proof}

\begin{remark} The bound obtained in Theorem~\ref{thm:twoagent} is tight, as demonstrated  by the instance given in Proposition~\ref{prop:identical}. In this instance,  there are two agents with identical subadditive valuations for which  no $\EFone$ allocation achieves  $\NSW$ more than a factor $\frac{1}{\sqrt{2}}$  of the optimal.
	
	\end{remark}

%% file: etc.tex
\section{Incompatibility Results}
{
As shown in Theorem~\ref{theorem:ef1-lowerb}, for subadditive valuations, there exists instances wherein no $\EFone$ allocation achieves $\NSW$ more than half of the of the optimal. This section refines the incompatibility to relevant special cases. 
 We first show that, even when agents have identical valuations, achieving $\EFone$ with an approximation better than $\frac{1}{\sqrt{2}}$ of the optimal $\NSW$ is not possible. This result also extends to partial and complete $\EFX$ allocations. Hence, the following result highlights the limitations of approaches explicitly designed for identical valuations, such as Leximin++ by \cite{PlautR20}.
}

	\begin{proposition}\label{prop:identical} There exists a fair division instance with identical, subadditive valutions such that no $\EFone$ allocation $\calA$ has Nash social welfare more than $\frac{1}{\sqrt{2}}$ times the optimal, i.e., for every $\EFone$ allocation $\calA$ we have 
		$$\NSW(\mathcal{A}) \leq \frac{1}{\sqrt{2}} \NSW(\mathcal{A}^*).$$
	 Here, $\calA^*$ denotes the Nash optimal allocation in the given instance.
		\end{proposition}
		\begin{proof}
			Consider a fair division instance with three goods $\{a,b,c\}$ and two agents with identical subadditive valuations given below: 
			$$
			v(S)=
			\begin{cases}
				2 \ \ \  &\text{if} \quad S=\{a,b,c\} \\
				2 \ \ \  &\text{if} \quad   S=\{b,c\} \\
				1 \ \ \ \ &\text{if} \quad  S=\{a,b\}  \text{ or } \{a,c\}  \\
				1 \ \ \ \ &\text{if} \quad  S=\{b\}  \text{ or } \{c\}  \\
				1-\epsilon &\text{if} \quad  S=\{a\}  \\
				0 \ \ \ \ &\text{if} \quad  S= \emptyset
			\end{cases}
			$$
			In this instance, the Nash optimal allocation  $\calA^*$  assigns one agent the set of goods $\{b,c\}$ and the other agent the good $\{a\}$.  This results in optimal Nash welfare of $\sqrt{2(1-\epsilon)}$. Note that the Nash optimal allocation is not  $\EFone$, since even after the removal of a good, the agent who got $\{a\}$ still envies the other agent. Observe that every allocation, excluding the Nash optimal one, has Nash welfare of at most $1$.  Hence, $\NSW(\calA) \leq 1$ for any $\EFone$ allocation $\calA$.  The stated bound in the proposition follows by letting $\epsilon \rightarrow 0$. 
			\end{proof}

\subsection{Incompatibility under Submodular Valuations}

Recall that submodular valuations constitute a sub-class of subadditive ones. The theorem below shows that, even for submodular (in fact, gross substitute) valuations, one cannot always expect $\EFone$ allocations with $\NSW$ more than {$\frac{1}{e^{1/e}} \approx \frac{1}{1.45}$} times the optimal.

\begin{theorem} For any $\varepsilon>0$, there exist fair division instances with submodular (in fact, gross substitute) valuations wherein no $\EFone$ allocation has Nash social welfare more than $\left(\frac{1}{e^{1/e}} +\varepsilon \right)$ times the optimal i.e., for \textit{every} $\EFone$ allocation $\calA$ 
we have
$$\NSW(\mathcal{A}) \leq \left( \frac{1}{e^{1/e}} +\varepsilon \right)    \NSW(\mathcal{A}^*).$$
Here, $\calA^*$ denotes the Nash optimal allocation in the given instance.
\end{theorem}

\begin{proof}
Consider division of $m$ indivisible goods among $n$ agents that are partitioned into two groups: $H$ and $L$. Agents $i \in H$ have (identical) valuation $v(S) \coloneqq |S|$, for all subsets $S \subseteq [m]$. In addition, each agent $j \in L$, has (identical) valuation $v'(S) \coloneqq |S|^{\delta}$, for all $S \subseteq [m]$ and a fixed parameter $\delta \coloneqq \frac{n}{|L| m} < 1$. 

Note that the agents valuations, $v$ and $v'$, are submodular (in fact, gross substitute) and they depend only on the cardinality of the assigned bundle; in this sense, the goods are identical. 

Here, the NSW of the optimal allocation $\mathcal{A}^*$ is at least as much as the $\NSW$ obtained by assigning $\frac{m-|L|}{|H|} $ goods to each agent $i \in H$, and one good to each agent $j \in L$.


 
 Further, assume that $m$ is divisible by $n$. In the current instance, in any $\EFone$ allocation $\calA=(A_1, \ldots, A_n)$ each agent must receive exactly $\frac{m}{n}$ goods. Note that, if $|A_i| \geq \frac{m}{n}+1$, for some $i$, then there must exist a bundle $A_j$ with cardinality at most $\frac{m}{n}-1$. In such a case, the $\EFone$ criterion would not hold for $j$ against $i$. Hence, in any $\EFone$ allocation $\calA$ we have $|A_i| = m/n$ for all $i$.
 
Therefore, for any $\EFone$ allocation $\calA$ the following bounds hold 
				\begin{align}
				\frac{\NSW(\calA)}{\NSW(\calA^*)} & \leq \left(\frac{  \left(\frac{m}{n} \right)^{\delta |L|}\left( \frac{m}{n} \right)^{|H|}}{\left( \frac{m-|L|}{|H|} \right)^{|H|}} \right)^{1/n} \nonumber \\ 
				& = \left(  \frac{m}{n} \right)^{ \frac{ \delta |L|}{n} } \left( \frac{m}{m-|L|}  \right)^{\frac{|H|}{n}}  \left(  \frac{|H|}{n}  \right)^{\frac{|H|}{n}} \nonumber \\ 
				& = \left(  \frac{m}{n} \right)^{ \frac{ \delta |L|}{n} } \left(  \frac{1}{1- \frac{|L|}{m} }  \right)^{\frac{|H|}{n}}  \left(  \frac{|H|}{n}  \right)^{\frac{|H|}{n}} \nonumber \\
				& = \left(  \frac{m}{n} \right)^{ \frac{ 1}{m} }  \left(  \frac{1}{1- \frac{|L|}{m} }  \right)^{\frac{|H|}{n}}  \left(  \frac{|H|}{n}  \right)^{\frac{|H|}{n}}
				 \label{ineq:nsw-up} 
			\end{align} 
Setting $|H|/n \approx e^{-1}$, gives us $\left( \frac{|H|}{n} \right)^{\frac{|H|}{n}} \approx \frac{1}{e^{1/e}}$. Further, for any $\varepsilon>0$, a sufficiently large $m$ ensures that the product of first two terms in the right-hand-side of equation (\ref{ineq:nsw-up}) satisfies $\left(  \frac{m}{n} \right)^{ \frac{ 1}{m} }  \left(  \frac{1}{1- \frac{|L|}{m} }  \right)^{\frac{|H|}{n}} \leq 1+\varepsilon$. 

Hence, for any $\varepsilon>0$, with $|H|/n \approx e^{-1}$ and a sufficiently large $m$, equation (\ref{ineq:nsw-up}) reduces to the stated bound: $ \frac{\NSW(\calA)}{\NSW(\calA^*)}  \leq \frac{1}{e^{1/e}} +\varepsilon$. The theorem stands proved. 
		\end{proof}

\subsection{Impossibility Result for Superadditive Valuations}		
In contrast to subadditive valuations, superadditive functions capture complementarity in preferences. Specifically, a set function $f: 2^{[m]} \mapsto \mathbb{R}_+$ is said to be superadditive if $f(S \cup T) \geq f(S) + f(T)$, for all subsets $S, T \subseteq [m]$.

The proposition below shows that---quite unlike the subadditive case---one cannot expect any compatibility between $\EFone$ and $\NSW$ for superadditive valuations. Specifically, here we show that, under superadditive valuations, the multiplicative gap between the $\NSW$ of $\EFone$ allocations and the optimal can be unbounded. 
		
\begin{proposition}
There exist fair division instances, with superadditive valuations, wherein the optimal $\NSW$ is positive, but for \textit{every} $\EFone$ allocation $\calA$ we have $\NSW(\calA) = 0$. 
\end{proposition}
		\begin{proof}
			Consider an instance with two agents ($n=2$) and four goods $\{a,b,c, d\}$. Agent 1 has monotone superadditive valuation
			$$
			v_1(S) \coloneqq
			\begin{cases}
				2 \ \ \ \text{if} \ \ \ \  |S| = 4 \ \\
				1 \ \ \ \text{if} \ \ \ \   |S|  = 3 \ \\
				0 \ \ \ \ \text{otherwise.}
			\end{cases}
			$$
			Agent 2 has additive valuation $v_2(S)=|S|$.  Observe that, in a Nash optimal allocation $\calA^*$, agent 1 receives any bundle of size three, and the last remaining good is allocated to agent 2, resulting in $\NSW(\calA^*)=1$. However, in every $\EFone$ allocation $\calA=(A_1, A_2)$, the second agent receives at least two items. Hence, $v_1(A_1)=0$, showing that $\NSW(\calA)=0$. 
		\end{proof}